\documentclass[twoside,11pt]{article}

\usepackage{jmlr2e}

\usepackage{microtype}
\usepackage{graphicx}
\usepackage{subfigure}
\usepackage{booktabs}

\usepackage{notations}
\usepackage{mathletters}

\usepackage{hyperref}

\usepackage[utf8]{inputenc} 
\usepackage[T1]{fontenc}    
\hypersetup{
    colorlinks = true,
    allcolors = blue
}
\usepackage{url}            
\usepackage{booktabs}       
\usepackage{amsfonts}       
\usepackage{nicefrac}       
\usepackage{microtype}      
\usepackage{times}
\usepackage[scaled=0.8]{beramono}
\usepackage{fancybox}
\usepackage{changepage}
\usepackage{enumitem}
\usepackage{comment}
\usepackage{notations}
\usepackage{stmaryrd}
\usepackage{bold-extra}
\usepackage{xcolor}

\usepackage{appendix} 

\usepackage{algorithm,algorithmic}

\usepackage{graphicx}
\graphicspath{{.}{figs/}}
\usepackage{sidecap}
\usepackage{amsmath,amsfonts,amssymb, tikz}
\usepackage{mathletters} 
\usetikzlibrary{arrows,chains,matrix,positioning,scopes}

\usepackage{mathtools}
\usepackage{xspace}

\renewcommand{\Pr}{\mathbb{P}}
\newtheorem{assumption}{Assumption}
\newcommand\scalemath[2]{\scalebox{#1}{\mbox{\ensuremath{\displaystyle #2}}}}

\newcommand{\beq}{\begin{equation}}
\newcommand{\eeq}{\end{equation}}
\newcommand{\beqa}{\begin{eqnarray}}
\newcommand{\eeqa}{\end{eqnarray}}
\newcommand{\beqan}{\begin{eqnarray*}}
	\newcommand{\eeqan}{\end{eqnarray*}}

\newcommand{\beqannumb}{\begin{eqnarray}}
\newcommand{\eeqannumb}{\end{eqnarray}}
\newcommand{\hl}[1]{{\textit{#1}}}

\usepackage[colorinlistoftodos, textwidth=4cm, shadow]{todonotes}
\newcommand{\oam}[1]{\todo[inline,color=orange!80]{{\it OAM:~}#1}}

\long\def\acks#1{\vskip 0.3in\noindent{\large\bf Acknowledgments}\vskip 0.2in
\noindent #1}

\firstpageno{1}

\begin{document}

\title{Optimal Strategies for Graph-Structured Bandits}

\author{\name Hassan Saber \email hassan.saber@inria.fr \\
       \addr SequeL Research Group \\
       Inria Lille-Nord Europe \& CRIStAL\\
       Villeneuve-d'Ascq, Parc scientifique de la Haute-Borne, France
       \AND
       \name Pierre Ménard \email pierre.menard@inria.fr \\
       \addr SequeL Research Group \\
       Inria Lille-Nord Europe \& CRIStAL\\
       Villeneuve-d'Ascq, Parc scientifique de la Haute-Borne, France
       \AND
       \name Odalric-Ambrym Maillard \email odalric.maillard@inria.fr \\
       \addr SequeL Research Group \\
       Inria Lille-Nord Europe \& CRIStAL\\
       Villeneuve-d'Ascq, Parc scientifique de la Haute-Borne, France}

\editor{}

\maketitle

\begin{abstract} We study a structured variant of the multi-armed bandit problem
 specified by a set of Bernoulli distributions $ \nu \!= \!(\nu_{a,b})_{a \in \cA, b \in \cB}$  with means $(\mu_{a,b})_{a \in \cA, b \in \cB}\!\in\![0,1]^{\cA\times\cB}$ and  by a given weight matrix $\omega\!=\! (\omegabb)_{b,b' \in \cB}$, where $ \cA$  is a finite set of arms and $ \cB $ is a finite set of users. The weight matrix $\omega$ is such that for any two users $b,b'\!\in\!\cB, \max_{a\in\cA}\abs{\muab \!-\! \mu_{a,b'}} \!\leq\! \omegabb $. 
 This formulation is flexible enough to capture various situations, from highly-structured scenarios  ($\omega\!\in\!\{0,1\}^{\cB\times\cB}$) to fully unstructured setups ($\omega\!\equiv\! 1$).
We consider two scenarios depending on whether the learner chooses only the actions to sample rewards from or both users and actions. We first derive problem-dependent lower bounds on the regret for this generic graph-structure that involves a structure dependent linear programming problem.
 Second, we adapt to this setting the Indexed Minimum Empirical Divergence (\IMED) algorithm introduced by Honda and Takemura (2015), and introduce the \IMEDSstar algorithm. Interestingly, \IMEDSstar does not require computing the solution of the linear programming problem more than about $\log(T)$ times after $T$ steps, while being provably asymptotically optimal. Also, unlike existing bandit strategies designed for other popular structures, \IMEDSstar does not resort to an explicit forced exploration scheme and only makes use of local counts of empirical events.
 We finally provide numerical illustration of our results that confirm the performance of \IMEDSstar. 
\end{abstract}

\section{Introduction}

The multi-armed bandit problem is a popular framework to formalize sequential decision making problems.
It was first introduced in the context of medical trials \citep{thompson1933likelihood,thompson1935criterion} and later formalized by \cite{ro52}.
In this paper, we consider a contextual and structured variant of the problem, specified by a set of distributions $ \nu \!=\! (\nu_{a,b})_{a \in \cA, b \in \cB}$ with means $(\mu_{a,b})_{a \in \cA, b \in \cB}$, where $ \cA $  is a finite set of arms and $ \cB $ is a finite set of users. Such $ \nu $ is called a (bandit) configuration where each $ \nu_b \!=\! (\nu_{a,b})_{ a \in \cA} $ can be seen  as a classical multi-armed bandit problem.
The streaming protocol is the following: at each time $ t \!\geq\! 1 $, the learner deals with a user $b_t \!\in\!\cB$ and chooses an arm $ a_t \!\in\! \cA $, based only on the past. We consider two scenarios: either  the sequence of users is deterministic (uncontrolled scenario) or the learner has the possibility to choose the user (controlled scenario), see Section~\ref{sec:setting}. The learner then receives and observes a reward $ X_t $ sampled according to $ \nu_{a_t,b_t} $ conditionally independent from the past. We assume binary rewards: each $ \nu_{a,b} $ is a Bernoulli distribution $ \textnormal{Bern}(\muab) $ with mean $ \muab \!\in\! (0,1)$ and we denote by $\cD$ the set of such configurations. The goal of the learner is then to maximize its expected cumulative reward over $T$ rounds, or equivalently minimize regret given by 
\[
R(\nu,T) = \Esp_\nu\!\brackets{\sum\limits_{t=1}^T \max\limits_{a \in \cA}\mu_{a,b_t} - X_t}  \,.
\]
For this problem one can run, for example, a separate instance of a bandit algorithm for each user $b$, but we would like to exploit a \hl{known structure} among the users (which we detail below).

\paragraph{Unstructured bandits}
The classical bandit problem (when $|\cB|\!=\!1$) received increased attention in the middle of the $20^{\text{th}}$ century. The seminal paper \cite{lai1985asymptotically} established the first  lower bounds on the cumulative regret, showing that designing a strategy that is optimal uniformly over a given set of configurations $\cD$ comes with a price. 
The study of the lower performance bounds in multi-armed bandits successfully lead to the development of asymptotically optimal strategies for specific configuration sets, such as the \KLUCB strategy \citep{lai1987adaptive,CaGaMaMuSt2013,maillard2018boundary} for exponential families, or alternatively the \DMED and \IMED strategies from \cite{honda2011asymptotically,honda2015imed}.
The lower bounds from \cite{lai1985asymptotically}, later extended by \cite{burnetas1997optimal} did not cover all possible configurations, and in particular \textit{structured} configuration sets were not  handled until \cite{agrawal1989asymptotically} and then \cite{graves1997asymptotically} established generic lower bounds. Here, structure refers to the fact that pulling an arm may reveal information that enables to refine estimation of other arms.
Unfortunately, designing efficient strategies that are provably optimal remains a challenge for many structures at the cost of a high computational complexity.

\paragraph{Structured configurations}
Motivated by the growing popularity of bandits in a number of industrial and societal application domains, the study of structured configuration sets has received increasing attention over the last few years:
The linear bandit problem is one typical illustration \citep{abbasi2011improved, srinivas2010gaussian, durand2017streaming}, for which the linear structure considerably modifies the achievable lower bound, see \cite{lattimore2017end}. 
The study of a unimodal structure naturally appears in the context of wireless communications, and has been considered in \cite{combes2014unimodal} from a bandit perspective, providing an explicit lower bound together with a strategy exploiting this structure.
Other structures  include Lipschitz bandits \citep{magureanu2014oslb}, and we refer to the manuscript \cite{magureanu2018efficient} for other examples, such as cascading bandits that are useful in the context of recommender systems.
\cite{combes2017minimal} introduced a generic strategy called \OSSB (Optimal Structured Stochastic Bandit), stepping the path towards generic multi-armed bandit strategies that are  adaptive to a given structure.

\paragraph{Graph-structure}
In this paper, we consider the following structure: For a \hl{given} weight matrix $\omega \!=\! (\omega_{b,b'})_{b,b'\in\cB}\in[0,1]^{\cB\times\cB}$ inducing a metric on $\cB$, we assume that for any two users $b,b'\!\in\!\cB$, $ \abs{\abs{ \mu_b \!-\! \mu_{b'} }}_\infty\!\coloneqq\!\max_{a \in \cA}\abs{\muab \!-\! \mu_{a,b'}} \!\leq\! \omegabb $.
We see the matrix $\omega$ as an adjacency matrix of a fully connected weighted graph where
each vertex represents a user and each weigh $\omega_{b,b'}$ measures proximity between two users, hence we call this a ``\hl{graph structure}''. The motivation to study such a structure is two-fold. On the one hand, in view of paving the way to  solving generic structured bandits, the graph structure yields nicely interpretable lower bounds that show how $\omega$ effectively modifies the achievable optimal regret and suggests a natural strategy, while being flexible enough to \hl{interpolate} between a fully unstructured and a highly structured setup. 
On the other hand, multi-armed bandits have been extensively applied to recommender systems: In such systems it is natural to assume that users may not react arbitrarily differently from each other, but that two users that are "close" in some sense will also react similarly when presented with the same item (action). Now, the similarity between any two users may be loosely or accurately known (by studying for instance activities of users on various social networks and refining this knowledge once in a while): The weight matrix $\omega$ enables to summarize such imprecise knowledge. Indeed $\omega_{b,b'}\!=\!0$ means that two users behave identically, while $\omega_{b,b'}\!=\!1$ is not informative on the true similarity $\|\mu_b\!-\!\mu_{b'}\|_\infty$ that can be anything from arbitrarily small to $1$.
Hence, studying this structure is both motivated by a theoretical challenge and more applied considerations. To our knowledge this is the first work on \textit{graph structure}. Other structured problems such as Clustered bandits \citep{gentile2014online}, Latent bandits \citep{maillard2014latent}, or Spectral bandits \citep{valko2014spectral} do not deal with this particular setting.

\paragraph{Goal}
The primary goal of this paper is to build a provably optimal strategy for this \hl{flexible} notion of structure. To do so, we derive lower bounds and use them to build intuition on how to handle structure, which enables us to establish a novel bandit strategy, that we prove to be optimal. Although specialized to this structure, the mechanisms leading to the strategy and introduced in the proof technique are novel and are of independent interest.

\paragraph{Outline and contributions} We formally introduce the graph-structure model in Section~\ref{sub:graphstructure}. 
Graph structure is simple enough while interpolating between a fully unstructured case and highly-structured settings such as clustered bandits (see Figure~\ref{fig:example}): This makes it a convenient setting to study \hl{structured} multi-armed bandits.
In Section~\ref{sec:lower_bounds}, we first establish in Proposition~\ref{prop:LB_pull}
a lower bound on  the asymptotic number of times a sub-optimal couple must be pulled by any consistent strategy (see Definition~\ref{def:consistent}),
together with its corresponding lower bound on the regret  (see Corollary~\ref{cor:LB_regret}) involving an optimization problem.
In Section~\ref{sec:imed_algo}, we revisit the Indexed Minimum Empirical Divergence (\IMED) strategy from \cite{honda2011asymptotically} introduced for unstructured multi-armed bandits, and adapt it to the graph-structured setting, making use of the lower bounds of Section~\ref{sec:lower_bounds}.
The resulting strategy is called \IMEDS in the controlled scenario and  \IMEDStwo in the uncontrolled scenario. Our analysis reveals that in view of asymptotic optimality, these strategies may still not optimally exploit the graph-structure in order to trade-off information gathering and low regret.
In order to address this difficulty, we introduce the modified \IMEDSstar strategy for the controlled scenario (and  \IMEDSstartwo in the uncontrolled one).
We show in Theorem~\ref{th:asymptotic_optimality_IMEDSstar}, which is the main result of this paper, that both \IMEDSstar and  \IMEDSstartwo are \hl{asymptotically optimal} consistent strategies.
Interestingly, \IMEDSstar does not compute a solution to the optimization problem appearing in the lower bound \emph{at each time step}, unlike for instance \OSSB introduced for generic structures, but only about $\log(T)$ times after $T$ steps. Also, if \textit{forced exploration} does not seem to be avoidable for this problem, \IMEDSstar does not make use of an explicit \emph{forced exploration} scheme 
but a more implicit one, based on local counters of empirical events. Up to our knowledge, \IMEDSstar is the first strategy with such properties, in the context of a structure requiring to solve an optimization problem, that is provably asymptotically optimal.
On a broader perspective, we believe the mechanism used in \IMEDSstar as well as the proof techniques could be extended beyond the considered graph-structure, thus  opening promising perspective in order to build structure-adaptive optimal strategies for generic structures. 
Last, we provide in Section~\ref{sec:numerical_experiments} numerical illustrations on synthetic data. They show that \IMEDSstar is also \hl{numerically efficient} in practice, both in terms of regret minimization and computation time; this contrasts with 
some bandit strategies introduced for other structures (as in \cite{combes2017minimal}, \cite{lattimore2017end}), that in practice suffer from a prohibitive burn-in phase.

\subsection{Setting}
\label{sec:setting}
Let us recall that the goal of the learner is to maximize its expected cumulative reward over $T$ rounds, or equivalently minimize regret given by 
\[ 
R(\nu,T) = \Esp_\nu\!\brackets{\sum\limits_{t=1}^T \max\limits_{a \in \cA}\mu_{a,b_t} - X_t}\,.
\]
As mentioned, for this problem one can run, for example,  a separate instance of  bandit algorithms for each user $b$, but we would like to exploit the graph structure. We consider two typical scenarios.

\paragraph{Uncontrolled scenario}  The sequence of users $(b_t)_{t \geq 1}$ is assumed deterministic and does not depend on the strategy of the learner. At each time step $t\!\geq\!1$, the user $b_{t}$ is revealed to the learner.

\paragraph{Controlled scenario} The sequence of users $(b_t)_{t \geq 1}$ is strategy-dependent and at each time step $t \!\geq\! 1$, the learner has to choose a user $b_t$ to deal with, based only on the past. 

~\\ Both scenarios are motivated by practical considerations: \textit{uncontrolled scenario} is the most common setup for recommender systems, while \textit{controlled scenario} is more natural in case the learner interacts actively with available users  as in advertisement campaigns. In an \textit{uncontrolled scenario}, the frequencies of user-arrivals are imposed and may be arbitrary, while in a \textit{controlled scenario} all users are available and 
the learner has to deal with them with similar frequency (even if this means considering a subset of users). We formalize the notion of frequency in the following definition.
\begin{definition}[Log-frequency of a user]\label{def:unifspread} 
A sequence of user $(b_t)_{t\geq1}$ has log-frequencies $\beta \!\in\! [0,1]^\cB $ if, almost surely,   the number of times the learner has dealt with user $b \!\in\!\cB$ is 
 $N_b(T) \!=\! \Theta\!\left(T^{\beta_b}\right)$\footnote{We say that $u_T = \Theta(v_T)$, if the two sequences $u_T$ and $v_T$ are equivalent.}. In this case, almost surely we have
\[
 \forall b \in \cB,\  \limT \dfrac{\log\!\left(N_b(T)\right)}{\log(T)} = \beta_b \,.
\]
\end{definition}
 In an \textit{uncontrolled scenario}, we assume that the sequence of users $(b_t)_{t \geq 1}$ has positive log-frequencies $\beta \!\in\! (0,1]^\cB $, with $\beta$ unknown to the learner. In a \textit{controlled scenario}, we focus only on strategies that induce sequences of users with
 same log-frequencies, hence all equal to $1$, independently on the considered configuration, that is strategies such that, almost surely, $N_b(T) \!=\! \Theta(T)$ for all user $b \!\in\!\cB$.
\subsection{Graph Structure}\label{sub:graphstructure}

In this section, we introduce the graph structure. We assume that all bandit configurations $\nu$ belong to a set of the form:
\[
\overline{\cD}_\omega \coloneqq \Set{\nu\in \cD :\ \forall b, b' \in \cB,\, \max\limits_{a \in \cA}\abs{\muab - \mu_{a,b'}} \leq \omega_{b,b'}  }\,,
\]
where $ \omega \!=\! (\omegabb)_{b,b' \in \cB} \in  [0,1]^{\cB \times \cB} $ is a weight matrix known to the learner. Intuitively, when the weights are  close to $1$, we expect no change to the agnostic situation. But, when the weights are close to $\|\mu_b\!-\!\mu_{b'}\|_\infty\!\coloneqq\!\max_{a\in\cA}\abs{\mu_{a,b} \!-\! \mu_{a,b'}}$, we expect significantly lower achievable regret.  

\begin{remark}
For the specific case where $ \omegabb=0$, $\overline{\cD}_\omega$ corresponds to user $b$ and $b'$ known to be perfectly clustered. 
The weight matrix given in Figure~\ref{fig:example} models three smooth clusters of users. Each cluster is included in a ball of diameter $\alpha$ for the infinite norm $\|\cdot\|_\infty$.
\end{remark}
In the sequel we assume the following properties on the weights.
\begin{assumption}[Metric weight property]\label{ass:metric}
The weight matrix $\omega$ satisfies:
\begin{itemize}[label={-}, itemsep = 0 mm]
    \item  $ \omega_{b,b} \!=\! 0 $ and $\omegabb \!>\! 0$ for all $b \!\neq\! b' \!\in\! \cB$,
    \item $ \omegabb \!=\! \omega_{b',b}$ and $\omegabb \!\leq\! \omega_{b,b''} \!+\! \omega_{b'',b'} $ for all  $ b, b', b'' \!\in\! \cB $.
\end{itemize}
\end{assumption}
This comes without loss of generality, since for the first property, if two users share exactly the same distribution we can see them as one unique user. For the second property, considering $ \omegatildbb \!=\! \sup_{a \in \cA, \nu \in \overline{\cD}_{\omega}}{\abs{\muab \!-\! \mu_{a,b'}}}  $ leads to the same set of configuration $\overline{\cD}_\omega\!=\! \overline{\cD}_{\omegatild} $ and it holds $ \omegatildbb \!=\! \omegatild_{b',b}$,\, $  \omegatildbb \!\leq\! \omegatild_{b,b''} \!+\! \omegatild_{b'',b'} $. Such a weight matrix $\omega$ naturally induces a metric on $\cB^2$.
\begin{figure}[H]
  \begin{minipage}{.5\textwidth}
\[
\left(\scalemath{0.70}{
\begin{array}{c c c | c c c | c c c}

 \multicolumn{1}{|c}{0} &  & (\alpha) &  &  & \mc{} &  &  &  \\ 
 \multicolumn{1}{|c}{} & \ddots &  &  &  & \mc{} &  & (1) &  \\ 
 \multicolumn{1}{|c}{(\alpha)} &  & 0 &  &  & \mc{} &  &  &  \\ 
\cline{1-6}
  \mc{} &  &  & 0  &  & (\alpha) &  &  &  \\ 
  \mc{} &  &  &   & \ddots &  &  &  &  \\ 
  \mc{} &  &  & (\alpha)  &  & 0 &  &  &  \\ 
\cline{4-9}
  \mc{} &  & \mc{} &   &  &  & 0 &  &  \multicolumn{1}{c|}{(\alpha)}\\ 
  \mc{} & (1) & \mc{} &   &  &  &  & \ddots &  \multicolumn{1}{c|}{}\\ 
  \mc{} &  & \mc{} &   &  &  & (\alpha) &  &  \multicolumn{1}{c|}{0}\\ 
\end{array}}
\right)
\]
  \end{minipage}%
  \begin{minipage}{.5\textwidth}
    \centering
    \includegraphics[scale=0.35]{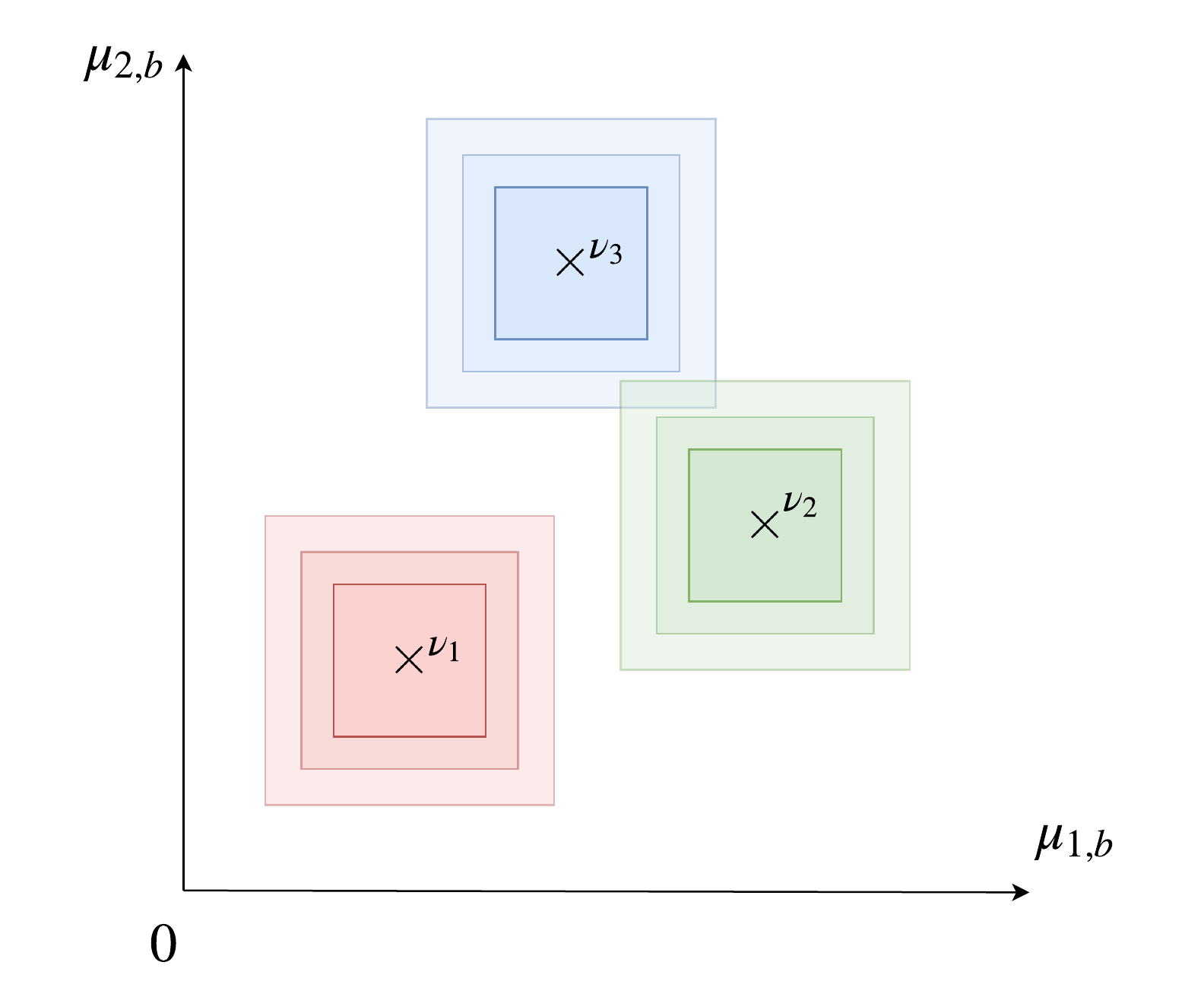}
  \end{minipage}
  \caption{A cluster structure. \textit{Left:} weight matrix of three clusters. \textit{Right:} range of two armed bandit problems included in clusters with center $(\nu_1,\nu_2,\nu_3)$ for various $\alpha$ (the larger $\alpha$ the lighter and larger the box). The value $\alpha\!=\!0$ corresponds to perfect clusters.}
  \label{fig:example}
\end{figure}

\subsection{Notations}
Let $\mu_b^\star \!=\! \max_{a \in \cA }{\muab} $ denote the optimal mean for user $b$ and $\cA_b^\star \!=\! \argmax_{a \in \cA}{\muab}$ the set of optimal arms for this user. We define for a couple $(a,b)\!\in\! \cA\!\times\!\cB$ its gap $\Delta_{a,b} \!=\! \mu^\star_b \!-\! \mu_{a,b}$. Thus a couple is optimal if its gap is equal to zero and sub-optimal if it is  positive. We denote by $\cO^\star \!=\! \Set{ (a,b) \!\in\! \cA\!\times\!\cB \!: \muab \!=\! \mu_b^\star  }$ the set of optimal couples.
Thanks to the chain rule we can rewrite the regret as follows:
\[
R(\nu,T) = \sum_{a,b \in \cA\times \cB} \deltaab\, \Esp_\nu\big[\Nab(T)\big]\,, \quad
\text{where }\Nab(t) = \sum_{s=1}^t \ind_{\big\{(a_s,b_s)=(a,b)\big\}}
\]
is the number of pulls of arm $a$ and user $b$ up to time $t$.

\section{Regret Lower bound}
\label{sec:lower_bounds}
In this subsection, we establish lower bounds on the regret for the structure $\overline{\cD}_\omega$. In order to obtain non trivial lower bounds we consider, as in the classical bandit problem, strategies that are \emph{consistent} (uniformly good) on $\overline{\cD}_\omega$.

\begin{definition}[Consistent strategy]\label{def:consistent}
A strategy is consistent on $\overline{\cD}_\omega$ if for all configuration $\nu\!\in\! \overline{\cD}_\omega$, for all sub-optimal couple $(a,b)$, for all $ \alpha \!>\!  0$,
\[
\limT\Esp_\nu \!\left[\dfrac{N_{a,b}(T)}{N_b(T)^\alpha}\right] = 0\,.
\]
\end{definition}
\begin{remark}
When $\cB \!=\! \{b\}$, $N_b(T) \!=\! T$ and we recover the usual notion of consistency \citep{lai1985asymptotically}.
\end{remark}
Before we provide below the lower bound on the cumulative regret, let us give some intuition:
To that end, we fix a configuration $\nu \!\in\! \overline{\cD}_\omega$  and a sub-optimal couple $(a,b)$. One key observation is that if for all $ b' \!\in\! \cB $ it holds $ \mu_b^\star \!-\! \mu_{a,b'} \!<\! \omega_{b,b'} $, this means we can form an environment 
$\widetilde \nu \!\in\! \overline{\cD}_\omega$ such that $\widetilde \mu_{a',b'}\!=\!\mu_{a',b'}$ for all couples $(a',b')$ except $(a,b)$, and such that $\widetilde \mu_{a,b}$  satisfies $\mu_b^\star \!<\! \widetilde \mu_{a,b} \!<\! \mu_{a,b'} \!+\! \omega_{b,b'}$. Indeed, in this novel environment,
$\widetilde \mu_{a,b} \!-\! \widetilde \mu_{a,b'}\!<\!\omega_{b,b'}$ still holds but $(a,b)$ is now optimal. Hence, we can transform the sub-optimal couple $(a,b)$ in an optimal one without moving the means of the other users. Thanks to this remarkable property, and introducing  $ \kl(\mu|\mu') $ to denote the Kullback-Leibler divergence between two Bernoulli distributions $\textnormal{Bern}(\mu)$ and $\textnormal{Bern}(\mu')$ with the usual conventions, one can prove then that for all consistent strategy
\[
\liminfT \Esp_\nu\!\left[\dfrac{\Nab(T)}{\log\!\left(N_b(T)\right)}\right] \geq \dfrac{1}{\klab}\,,
\]
which is the lower bound that we get \emph{without} graph structure. This suggests that only the users $ b' $ such that $ \mu_b^\star \!-\! \mu_{a,b'} \!>\! \omega_{b,b'} $ provide information about the behavior of user $ b $.
This justifies to introduce for each couple $(a,b)$ the fundamental set
 \[
 \Bab \coloneqq \Set{b' \in \cB:\ \mu_{a,b'} < \mu_b^\star - \omegabb}\,.
 \]
It is also convenient to introduce its frontier, denoted
 $\partial \cB_{a,b}\!\coloneqq\!\Set{b' \!\in\! \cB\!: \mu_{a,b'} \!=\! \mu_b^\star \!-\! \omegabb}$.
Now, in order to report the lower bounds  while avoiding tedious technicalities, we slightly restrict the set $\overline{\cD}_\omega$. To this end, we introduce  the set
\beqan
\cD_\omega \coloneqq  \Set{\nu \in \overline{\cD}_\omega:\ \forall (a,b) \in \cA\times\cB,\, \partial\cB_{a,b}=\emptyset} \,.
\eeqan
This definition is justified since the closure of $\cD_\omega $ is indeed $\overline{\cD}_\omega$ (we only remove from $\overline{\cD}_\omega$ sets of empty interior). We can now state the following proposition.
\begin{proposition}[Graph-structured lower bounds on pulls]\label{prop:LB_pull} Let us consider a consistent strategy. Then, for all configuration $ \nu\!\in\!\cD_\omega$, almost surely it holds for all sub-optimal couple $(a,b) \!\notin\!\cO^\star$,
\beq 
\label{eq:lb_couple}
\limT \!N_b(T) \!<\! + \infty \quad \textnormal{or} \quad 
\liminfT{ \frac{1}{\log\!\left(N_b(T)\right)} \sum\limits_{b'\in \Bab}\!{\klof{\mu_{a,b'}}{\mu_b^\star - \omegabb} \!\, N_{a,b'}(T)}} \geq 1\,.
\eeq
\end{proposition}
We then introduce the notion of \textit{Pareto-optimality} based on the lower bounds given in  Proposition~\ref{prop:LB_pull}.
\begin{definition}[Pareto-optimality]
A strategy is asymptotically Pareto-optimal if for all $\nu \!\in\! \cD_\omega$, 
\[
\forall a\in\cA, \quad \limsupT \min\limits_{b:\, (a,b) \notin \cO^\star}\frac{1}{\log\!\left(N_b(T)\right)} \sum\limits_{b'\in \Bab}\!\klof{\mu_{a,b'}}{\mu_{b}^\star - \omega_{b,b'} }\!\, N_{a,b'}(T) \leq  1\,,
\]
with the convention $\min_\emptyset \!=\! -\infty$.
\end{definition}

\begin{remark}
This  proposition reveals that the set 
$\Bab\!=\!\Set{b' \!\in\! \cB\!: \mu_{a,b'} \!<\! \mu_b^\star \!-\! \omegabb}$ plays a crucial role in the graph structure.
The definition of $\cD_\omega$ excludes specific situations when there exists $b,b'\!\in\!\cB$, $a\!\in\!\cA$, $ \omega_{b,b'}\!=\!\mu_{b}^\star\!-\!\mu_{a,b'}\!=\! \deltaab\!+\!\mu_{a,b}\!-\!\mu_{a,b'}$, that belong to the close set $\overline{\cD}_\omega$. Extending the result to $\overline{\cD}_\omega$ seems possible but at the price of clarity due to the need to handle degenerate cases.
\end{remark}
In order to derive an asymptotic lower bound on the regret from these asymptotic lowers bounds, we have to characterize the growth of the counts $\left(N_b(\cdot)\right)_{b \in \cB}$. 
\begin{corollary}[Lower bounds on the regret]\label{cor:LB_regret}Let us consider a consistent strategy and sequences of users with log-frequencies $\beta \!\in\! (0,1]^\cB$ independently of the considered configuration in $\cD_\omega$. Then, for all configuration $ \nu \!\in\! \cD_\omega$ 
\beqannumb
\liminfT \dfrac{R(\nu,T)}{\log(T)} \geq C_\omega^\star(\beta,\nu):=
& \min&  \bigg\{\sum\limits_{a,b \notin \cO^\star}  \deltaab\,n_{a,b}:\ n \in \Real_+^{\cA\times\cB}  \label{eq:lb_regret_structure}
\\
& s.t. & \forall (a,b) \notin \cO^\star, \quad \sum\limits_{b' \in \Bab}\!\klof{\mu_{a,b'}}{\mu_b^\star - \omegabb}\!\, n_{a,b'} \geq \beta_b\,\bigg\}\,. \nonumber
\eeqannumb
\end{corollary}
Hence such a strategy is asymptotically optimal if for all $\nu \!\in\! \cD_\omega$
\[
\limsupT \dfrac{R(\nu,T)}{\log(T)} \leq C_\omega^\star(\beta,\nu)\,.
\]
\begin{remark}In the previous corollary, log-frequencies $\beta$ may be either strategy dependent or independent. In an \textit{uncontrolled scenario}, $\beta$ is imposed by the setting and does not depend on the followed strategy, while in a \textit{controlled scenario} we consider strategies that impose $\beta \!=\! 1_\cB \!\coloneqq\! (1)_{b \in\cB} $. 
\end{remark}
Like other structured bandit problems (as in \cite{combes2017minimal}, \cite{lattimore2017end}) this lower bound is characterized by a problem-dependent constant $C_\omega^\star(\beta,\nu)$ solution to an \hl{optimization problem}.
In the agnostic case we recover the lower bound of the classical multi-armed bandit problem. Indeed, let us introduce for $\alpha\!\in\![0,1]$ the weight matrix $\omega_\alpha$ where all the weights are equal to $\alpha$ (except for the zero diagonal). $\omega_\alpha$ is the same weight matrix as in Figure~\ref{fig:example} but only for one cluster. Then  when there is no structure  ($\omega\!\equiv\!\omega_1$), we obtain the explicit constant
\begin{equation}
\label{eq:constant_agnostic}
    C_{\omega_1}^\star(\beta,\nu) = \sum_{b \in \cB}\beta_b\!\! \sum_{a \in \cA:\ (a,b)\notin \cO^\star}\frac{\deltaab}{\kl(\mu_{a,b}|\mu_b^\star)}\,,
\end{equation}
that corresponds to solving $|\cB|$ bandit problems in parallel (independently the ones from the others). Thus the graph structure allows to interpolate smoothly between $|\cB|$ independent bandit problems and a unique one when all the users share the same distributions. 
In order to illustrate the gain of information due to the graph structure we plot in Figure~\ref{fig:quotient_c} the expectation $\Esp_{\nu \sim \mathcal{U}(\cD_{\omega_\alpha})}\!\left[ C_{\omega_\alpha}^\star(1_\cB,\nu)/C_{\omega_1}^\star(1_\cB,\nu)\right]$ of the ratio between the constant in the structured case~\eqref{eq:lb_regret_structure} and that in the agnostic case~\eqref{eq:constant_agnostic}, where $\mathcal{U}(\cD_{\omega_\alpha})$ denotes the uniform distribution over $\cD_{\omega_\alpha}$, $\alpha \!\in\! [0,1]$. 
\begin{figure}[H]
    \centering
    \includegraphics[scale=0.7]{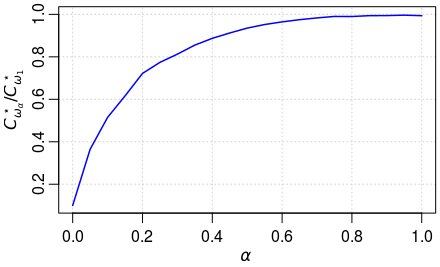}
        \caption{Plot of $\alpha \mapsto \Esp_{\nu \sim 
        \mathcal{U}(\cD_{\omega_\alpha})}\!\left[ C_{\omega_\alpha}^\star(1_\cB,\nu)/C_{\omega_1}^\star(1_\cB,\nu)\right]$  where $\omega_\alpha$ is a matrix where all the weights are equal to $\alpha$ (except for the zero diagonal) and $\nu$ is sampled uniformly at random in $\cD_{\omega_\alpha}$.  }
    \label{fig:quotient_c}
\end{figure}


\section{\IMED type strategies for Graph-structured Bandits}
\label{sec:imed_algo}
In this section, we present for both the controlled and uncontrolled scenarios, two strategies: \IMEDSstar that matches the asymptotic lower bound of Corollary~\ref{cor:LB_regret} and \IMEDS with a lower computational complexity but weaker guaranty. Both are inspired by the Indexed Minimum Empirical Divergence (\IMED) proposed by \citet{honda2011asymptotically}. The general idea behind this algorithm is to enforce, via a well chosen index, the constraints~\eqref{eq:lb_couple} that appears in the optimization problem~\eqref{eq:lb_regret_structure} of the asymptotic lower bound. These constraints intuitively serve as tests to assert whether or not a couple is optimal. 
\subsection{\IMED type strategies for the \textit{controlled scenario}}
We consider the \textit{controlled scenario} where the sequence of users $(b_t)_{t \geq 1}$ is strategy-dependent and at each time step $t \!\geq\! 1$, the learner has to choose a user $b_t$ and an arm $a_t$, based only on the past.

\subsubsection{The \texorpdfstring{\IMEDS}{TEXT}
strategy.}
We denote by $ \muhatab(t) \!=\! \frac{1}{N_{a,b}(t)} \sum\limits_{s=1}^t{\ind_{\{ (a_s,b_s)= (a, b)\}} X_s } $ if $ \Nab(t) \!>\!0  $, $ 0 $ otherwise,  the empirical mean of the rewards from couple $(a,b)$. Guided by the lower bound~\eqref{eq:lb_couple} we generalize the \IMED index to take into account the graph structure as follows. For a couple $(a,b)$ and at time $t$ we define
\begin{equation}
    \label{eq:def_imed_index_couple}
    I_{a,b}(t) = \Bigg\{\begin{array}{ll}
     \log\!\left(N_{a,b}(t)\right) &\text{if } (a,b) \in \Ohat^\star(t) \\ 
     \sum\limits_{b'\in \Bhat_{a,b}(t)} \! \klof{\muhat_{a,b'}(t)}{\muhat_{b}^\star(t) - \omegabb} \!\,N_{a,b'}(t)  + \log\!\left(N_{a,b'}(t)\right)  &\text{otherwise} 
\end{array} \,,
\end{equation}
where $\muhat_b^\star(t)\! =\! \max_{a\in\cA}\muhatab(t)$ is the current best mean for user $b$, the current set of optimal couple is
\[
\Ohat^\star(t) \coloneqq \bigg\{ (a,b) \in \cA \times \cB:\ \muhatab(t) = \muhat_b^\star(t)   \bigg\}
\]
and the current set of informative users for an empirical sub-optimal couple $(a,b)$ is 
\[
\Bhat_{a,b}(t) \coloneqq \bigg\{b' \in \cB:\ N_{a,b'}(t) > 0 \ad \muhat_{a,b'}(t) < \muhat_b^\star(t) - \omegabb \bigg\} \,.
\]
This quantity can be seen as a transportation cost for ``moving\footnote{This notion refers to the generic proof technique used to derive regret lower bounds. It involves a change-of-measure argument, from the initial configuration in which the couple is sub-optimal to another one chosen to make it optimal.}'' a sub-optimal couple to an optimal one, plus exploration terms (the logarithms of the numbers of pulls). When an optimal couple is considered, the transportation cost is null and only the exploration part remains. Note that, as stated in \citet{honda2011asymptotically}, $I_{a,b}(t)$ is an index in the weaker sense since it is not determined only by samples from the couple $(a,b)$ but also uses empirical means of current optimal arms. We define \IMEDS (Indexed Minimum Empirical Divergence for Graph Structure) to be the strategy consisting of pulling a couple with minimum index in Algorithm~\ref{alg:IMEDS}. It works well in practice, see Section~\ref{sec:numerical_experiments}, and has a low computational complexity (proportional to the number of couples). However, it is known for other structures, see \citet{lattimore2017end}, that such greedy strategy \hl{does not exploit optimally} the structure of the problem. Indeed, at a high level, pulling an apparently sub-optimal couple $(a,b)$ allows to gather information not only about this particular couple but also about other couples due to the structure. In order to attain optimality one needs to find  couples that provide the best trade-off between information and low regret. This is exactly what is done in the optimization problem~\eqref{eq:lb_regret_structure}.

\subsubsection{The \texorpdfstring{\IMEDSstar}{TEXT} strategy}\label{sub:IMEDstar}

In order to address this difficulty we first, thanks to the (weak) indexes, decide whether we need to exploit or explore. 
In the second case, in order to explore optimally according to the graph structure we  solve the optimization problem~\eqref{eq:lb_regret_structure} parametrized by the current estimates of the means and then track the optimal numbers of pulls given by the solution of this problem. More precisely at each round we \hl{choose} but not immediately pull a couple with minimum index
\[
(\aul_t,\bul_t) \in \argmin\limits_{(a,b)\in \cA\times\cB}I_{a,b}(t)\,.
\]
\textbf{Exploitation:} If this couple is currently optimal, $(\aul_t,\bul_t) \!\in\! \Ohat^\star(t)$, we exploit, that is pull this couple. \\
\textbf{Exploration:} Else we explore arm $a_{t+1}\!=\!\aul_t$. To this end, let $\nopt(t)$ be a solution of the empirical version of~$\eqref{eq:lb_regret_structure}$ with $\beta = 1_\cB$, that is
\beqa \label{eq:current_optimization_pb}
 \nopt(t) \in &\argmin\limits_{n \in \Real_+^{\cA\times\cB} }& \bigg\{ \sum\limits_{(a,b) \in \cA\times\cB}  \big(\muhat_b^\star(t)-\muhat_{a,b}(t)\big)\,n_{a,b}    \\
& s.t. & \forall (a,b) \notin \Ohat^\star(t), \Bhat_{a,b}(t) \neq \emptyset: \! \sum\limits_{b' \in \Bhat_{a,b}(t)}\klof{\muhat_{a,b'}(t)}{\muhat_b^\star(t) - \omegabb}\!\, n_{a,b'} \geq 1\bigg\}. \nonumber 
\eeqa
The current optimal numbers of pulls given by 
\beq \label{eq:Nopt}
\Nopt_{a,b}(t)  = \nopt_{a,b}(t) \min\limits_{b' \in \cB} I_{a,b'}(t) \,.
\eeq
We then \textit{track}
\beq \label{tracked user}
b_{t+1} \in \argmax\limits_{b \in \Bhat_{{\aul_t},{\bul_t}}(t)\cup\Set{\bul_t}} \Nopt_{\aul_t,b}(t) - N_{\aul_t,b}(t) \,.
\eeq
Asymptotically, we expect that all the sub-optimal couples are pulled roughly $\log(T)$ times. Therefore, for all sub-optimal couple $(a,b)$, the index $I_{a,b}(T)$ should be of order $\log(T)$. Thus we asymptotically recover in the definition of $\Nopt_{a,b}(\cdot)$ the optimal number of pulls of couple $(a,b)$, that is  $n^\nu_{a,b} \log(T)$ as  suggested in Corollary~\ref{cor:LB_regret}. Finally we pull the selected couple $(a_{t+1}, b_{t+1})$. In order to ensure optimality, however, such a direct tracking of the current optimal number of pulls is still a bit too aggressive and we need to force exploration in \textit{some} exploration rounds. We proceed as follows: when we explore arm $\aul_t$ we automatically pull a couple $(\aul_t,b)$ if its number of pulls $N_{\aul(t),b}$ is lower than the logarithm of the number of time we decided to explore this arm. See Algorithm~\ref{alg:IMEDSstar} for details. This does not hurt the asymptotic optimally because we expect to explore a sub-optimal arm not more than $\log(T)$ times. On the bright side, this is still different than the traditional forced exploration. Indeed, only few rounds are dedicated to exploration thanks to the first selection with the indexes and among them only a logarithmic number will consist of pure exploration: Thus, we expect an overall $\log\log(T)$ rounds of forced exploration. Note also that all the quantities involved in this forced exploration use empirical counters. Putting all together we end up with strategy \IMEDSstar described in Algorithm~\ref{alg:IMEDSstar}.
\paragraph{Comparison with other strategies}\IMEDSstar combines ideas from \IMED introduced by \citet{honda2011asymptotically} and from \OSSB by \citet{combes2017minimal}. More precisely, it generalizes the index from \IMED to the graph structure. From \OSSB it borrows the tracking of the optimal counts given by the asymptotic lower bound (see also \citet{lattimore2017end}) and the way to force exploration sparingly. The main difference with \OSSB is that \IMEDSstar leverages the indexes to deal with the exploitation-exploration trade-off. In particular \IMEDSstar \hl{does not need} to solve at each round the optimization problem~\eqref{eq:lb_regret_structure}. This greatly improves the computational complexity. Also, note that \OSSB requires choosing a tuning parameter that must  be positive to ensure theoretical guarantees but that must be set equal to $0$ to work well in practice. This is not the case for \IMEDSstar that requires no parameter tuning and that works well both in theory and in practice (see Section~\ref{sec:numerical_experiments}).
\begin{figure}[htb]
\centering
\makebox[0pt][c]{%
\begin{minipage}[]{0.45\textwidth}
\small\vspace{-5mm}
\begin{algorithm}[H]
\small
\caption{\IMEDS (controlled scenario)}
\label{alg:IMEDS}
\begin{algorithmic}
\REQUIRE Weight matrix $(\omegabb)_{b,b' \in \cB}  $.
\FOR {$ t = 1 ... T $} 
\STATE Pull $ (a_{t+1},b_{t+1}) \!\in\! \argmin\limits_{ (a,b)\in \cA\times\cB }\!I_{a,b}(t) $
\ENDFOR
\vspace{62.5 mm}
\end{algorithmic}
\end{algorithm}
\end{minipage}
\hfill
\begin{minipage}[]{0.57\textwidth}
\vspace{-5mm}
\begin{algorithm}[H]
\small
\caption{\IMEDSstar (controlled scenario)}
\label{alg:IMEDSstar}
\begin{algorithmic}
\REQUIRE Weight matrix $(\omegabb)_{b,b' \in \cB}  $.
\STATE $\forall a \in \cA, \quad c_a, c_a^+ \leftarrow 1$
\FOR{ For $ t = 1 ... T $}
\STATE Choose $ (\aul_t,\bul_t) \in \argmin\limits_{(a,b)\in \cA\times\cB}I_{a,b}(t) $
\IF{$(\aul_t,\bul_t)\in \Ohat^\star(t)$}
\STATE Choose $ (a_{t+1},b_{t+1}) =  (\aul_t,\bul_t)$
\ELSE
\STATE  Set $a_{t+1} = \aul_t$
\IF{$c_{a_{t+1}} = c_{a_{t+1}}^+ $}
\STATE $c_{a_{t+1}}^+ \leftarrow 2 c_{a_{t+1}}^+   $
\STATE Choose  $b_{t+1} \in \argminb \Nab(t) $ \vspace{-2mm}
\ELSE 
\STATE Choose  $b_{t+1} \!\!\in\!\!\!\!\! \argmax\limits_{b \in \Bhat_{{\aul_t},{\bul_t}}(t)\cup\Set{\bul_t}}\!\!\!\!\!\! \Nopt_{a_{t+1},b}(t) \!-\! N_{a_{t+1},b}(t)  $ \vspace{-2mm}
\ENDIF
\STATE $c_{a_{t+1}} \leftarrow c_{a_{t+1}}  + 1 $
\ENDIF
\STATE Pull $(a_{t+1},b_{t+1})$
\ENDFOR
\end{algorithmic}
\end{algorithm}
\end{minipage}
}%
\vspace{-5mm}
\end{figure}

\subsection{\IMED type strategies for the \textit{uncontrolled scenario}}
In this section, an \textit{uncontrolled scenario} is considered where the sequence of users $(b_t)_{t \geq 1}$ is assumed deterministic and does not depend on the strategy of the learner. We adapt the two previous strategies \IMEDS and \IMEDSstar to this scenario.

\paragraph{\IMEDStwo strategy} At time step $t\!\geq\!1$ the choice of user $b_{t}$ is no longer strategy-dependent but is imposed by the sequence of users $(b_t)_{t\geq1}$ which is assumed to be deterministic in the \textit{uncontrolled scenario}. The learner only chooses an arm to pull $a_{t}$ knowing user $b_{t}$.  We define \IMEDStwo to be the strategy consisting of pulling an arm with minimum index in Algorithm~\ref{alg:IMEDS2} of Appendix~\ref{app:algo_senario2}. \IMEDStwo suffers the same advantages and shortcomings as \IMEDS. It does not exploit optimally the structure of the problem but it works well in practice, see Section~\ref{sec:numerical_experiments}, and has a low computational complexity.

\paragraph{\IMEDSstartwo strategy} In order to explore optimally according to the graph structure in the \textit{uncontrolled scenario}, we also track the optimal numbers of pulls. $\beta$ may be at first glance different from $1_\cB$. This requires some normalizations.  First, for all time step $ t \!\geq\!1$, $\nopt(t)$ now denotes a solution of the empirical version of~$\eqref{eq:lb_regret_structure}$ with $\beta \!=\!(\hat \beta_b(t))_{b \in \cB}$ where $\hat\beta_b(t)\!=\! \!\log\!\left(N_b(t)\right)\!/\log(t)$ estimates log-frequency $\beta_b$ of user $b \!\in\!\cB$. Second, we have to consider normalized indexes $\widetilde I_{a,b}(t) \!=\! I_{a,b}(t)/\hat\beta_b(t)$ for couples $(a,b)\!\in\!\cA\!\times\!\cB $ in order to have $\widetilde I_{a,b}(T) \!\sim\! \log(T)$ as in
the \textit{controlled scenario}.  An additional difficulty is that at a given time step $t\!\geq\!1$, while the indexes indicate to explore, the current tracked user (see Equation~\ref{tracked user}) given is likely to be different from user $b_{t}$ with whom the learner deals. This difficulty is easy to circumvent by postponing and prioritizing the exploration until the learner deals with the tracked user. Priority in exploration phases is given to first delayed forced-exploration and delayed exploration based on solving optimization problem~\eqref{eq:lb_regret_structure}, then exploration based on current indexes (see Algorithm~\ref{alg:IMEDSstar2} in Appendix~\ref{app:algo_senario2}). \IMEDSstartwo corresponds essentially to \IMEDSstar with some delays due to the fact that the tracked and the current users may be different. This has no impact on the optimality of \IMEDSstartwo since log-frequencies of users are enforced to be positive.

\subsection{Asymptotic optimality  of IMED type strategies}
In order to prove the asymptotic optimality of \IMEDSstar we introduce the following mild assumptions on the configuration considered.
\begin{definition}[Non-peculiar configuration]
\label{def:not_pathological}
    A configuration $\nu\!\in\! \cD_\omega$ is non-peculiar if the optimization problem~\eqref{eq:lb_regret_structure} admits a unique solution and each user $ b $ admits a  unique optimal arm $ a_b^\star$.
\end{definition}
In Theorem~\ref{th:asymptotic_optimality_IMEDSstar} we state the main result of this paper, namely, the asymptotic optimality of \IMEDSstar and \IMEDSstartwo. We prove this result for \IMEDSstar in Appendix~\ref{app:proof_main_result} and adapt this proof in Appendix~\ref{app: extended proof} for \IMEDSstartwo. Please refer to Proposition~\ref{prop:upper bounds} (Appendix~\ref{app: IMEDSstar finite time analysis}) for more refined finite-time upper bounds. As a byproduct of this analysis we deduce the Pareto-optimality of \IMEDS and \IMEDStwo stated in Proposition~\ref{th:asymptotic_Pareto_optimality} and proved in Appendix~\ref{app: extended proof}.  


\begin{theorem}[Asymptotic optimality]
\label{th:asymptotic_optimality_IMEDSstar}
Both \IMEDSstar and \IMEDSstartwo are consistent strategies. Further, they are asymptotically optimal on the set of non-peculiar configurations, that is, for all $\nu \in \cD_\omega$ non-peculiar, under \IMEDSstar
the sequence of users has log-frequencies $1_\cB$ and 
we have
\[
\limsupT \dfrac{R(\nu,T)}{\log(T)} \leq C_\omega^\star(1_\cB,\nu)\,,
\]
and, under \IMEDSstartwo, assuming a sequence of users with log-frequencies $\beta \!\in\! (0,1]^\cB$, we have 
\[
\limsupT \dfrac{R(\nu,T)}{\log(T)} \leq C_\omega^\star(\beta,\nu)\,.
\]
\end{theorem}

\begin{proposition}[Asymptotic Pareto-optimality]
\label{th:asymptotic_Pareto_optimality}
Both \IMEDS and \IMEDStwo are consistent strategies. Further, they are asymptotically Pareto-optimal on the set of non-peculiar configurations, that is, under \IMEDS or \IMEDStwo, for all $\nu \in \cD_\omega$ non-peculiar,
\[
\forall a \in\cA, \quad \limsupT \min\limits_{b:\ (a,b) \notin \cO^\star}\frac{1}{\log\!\left(N_b(T)\right)} \sum\limits_{b'\in \Bab}\kl(\mu_{a,b'}|\mu_{b}^\star - \omega_{b,b'} ) N_{a,b'}(T) \leq  1\,.
\]
\end{proposition}


\paragraph{Discussion} Removing forced exploration remains the most challenging task for structured bandit problems. 
In the context of this structure, forced exploration would have been to use criteria like "if $N_{a,b}(T) \!<\! \hat{c}_{a,b}(T) \log(T)$, then pull couple $(a,b)$" for some constants $\hat{c}_{a,b}(T)$ that depends on the minimization problem coming from the lower bound and where $\hat{c}_{a,b}(T) \log(T)$ can be interpreted as an estimator of the theoretical asymptotic lower bound on the numbers of pulls of couple $(a,b)$. In stark contrast, in \IMEDSstar there is no forced exploration in choosing the  arm to explore and, in choosing the user to explore, the used criteria is more intrinsic as it reads "if $N_{a,b}(T) \!<\! \hat{c}_{a,b}(T) I_a(T) $, then pull couple $(a,b)$", where $I_{a}(T)\!\sim\! \log(T) $ but really depends on $(N_{a,b}(T))_{(a,b) \notin \mathcal{C}^\star}$. Thus, the used criteria are not asymptotic, and do not dependent on the time $t$ but on the current numbers of pull of sub-optimal arms. Since theoretical asymptotic lower bounds on the numbers of pulls are  significantly larger than the current numbers of pulls in finite horizon (see Figure~\ref{fig:experiments}), \IMEDSstar strategy is
also expected to behave better than strategies based on usual (conservative) forced exploration.  Although entirely removing forced exploration would be nicer, in \IMEDSstar, forced exploration is only done in a sparing way.

\section{Numerical experiments}
\label{sec:numerical_experiments}
In this section, we compare empirically the following strategies introduced beforehand: \IMEDS and \IMEDSstar
described respectively in Algorithms~\ref{alg:IMEDS},\,\ref{alg:IMEDSstar}, \IMEDStwo and \IMEDSstartwo described respectively in Algorithms~\ref{alg:IMEDS2},\,\ref{alg:IMEDSstar2} and the baseline \IMED by \citet{honda2011asymptotically} that does not exploit the structure. We compare these strategies on two setups, each with $|\cB|\!=\! 10$ users and $|\cA|\! =\!5$ arms. For the \textit{uncontrolled scenario} we consider the round-robin sequence of users. As expected the  strategies leveraging the graph structure perform better than the baseline \IMED that does not exploit it. Furthermore, the plots suggest that 
 \IMEDS and \IMEDSstar (respectively  \IMEDStwo and \IMEDSstartwo) perform similarly in practice.

\paragraph{Fixed configuration }Figure~\ref{fig:experiments}: \textit{Left} -- For these experiments  we investigate these strategies on a \emph{fixed} configuration. The weight matrix $\omega$ and the configuration $\nu\!\in\! \cD_\omega$ are given in Appendix~\ref{app:detail_numerical_exp}. This enables us to plot also the asymptotic lower bound on the regret for reference: We plot the unstructured lower bound (\verb|LB_agnostic|) in dashed red line, and the structured lower bound (\verb|LB_struct|) in dashed blue line.

\paragraph{Random configurations }Figure~\ref{fig:experiments}: \textit{Right} -- In these experiments we average regrets over random configurations. We proceed as follows: At each run we sample uniformly at random  
a weight matrix $\omega$  and then sample uniformly at random a configuration $\nu \!\in\!\cD_\omega$. 
\begin{figure}[H]
    \centering
     \includegraphics[scale=0.62]{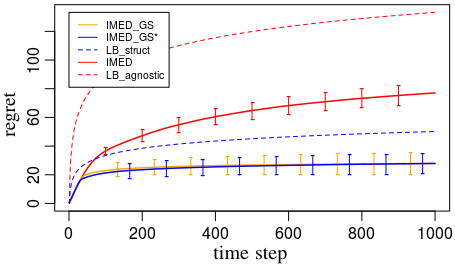} 
     \hfill \includegraphics[scale=0.62]{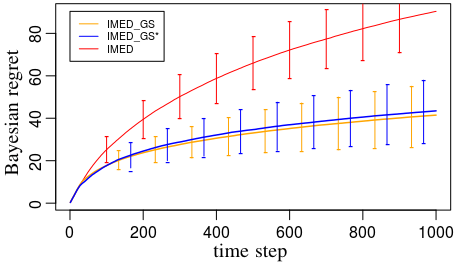}
     \hfill
     \includegraphics[scale=0.62]{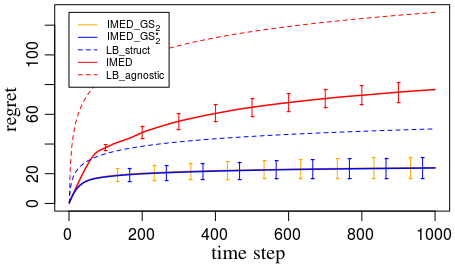} 
     \hfill\includegraphics[scale=0.62]{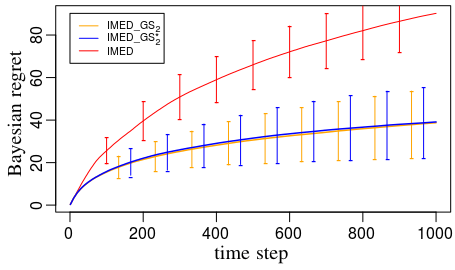}
\vspace{-5mm}\noindent
     \caption{  Regret approximated over $1000$ runs. \textit{Top:} \textit{controlled scenario}. \textit{Bottom:}  \textit{uncontrolled scenario},  $(b_t)_{t\geq1}$ is the round-robin sequence of users. \textit{Left:} Fixed configuration. \textit{Right:} Random configurations.}
 \label{fig:experiments}
\end{figure}

\vspace{-5mm}\noindent

~\\
Additional experiments in Appendix~\ref{app:detail_numerical_exp} confirm that both \IMEDS and \IMEDSstar induce sequences of users with log-frequencies all equal to $1$.

\acks{This work was supported by CPER Nord-Pas de Calais/FEDER DATA Advanced data science and
technologies 2015-2020, the French Ministry of Higher Education and Research, the French Agence
Nationale de la Recherche (ANR) under grant ANR-16- CE40-0002 (project BADASS), and Inria
Lille – Nord Europe.}

\newpage

\bibliography{biblio}

\begin{thebibliography}{23}
\providecommand{\natexlab}[1]{#1}
\providecommand{\url}[1]{\texttt{#1}}
\expandafter\ifx\csname urlstyle\endcsname\relax
  \providecommand{\doi}[1]{doi: #1}\else
  \providecommand{\doi}{doi: \begingroup \urlstyle{rm}\Url}\fi

\bibitem[Abbasi-Yadkori et~al.(2011)Abbasi-Yadkori, P{\'a}l, and
  Szepesv{\'a}ri]{abbasi2011improved}
Yasin Abbasi-Yadkori, D{\'a}vid P{\'a}l, and Csaba Szepesv{\'a}ri.
\newblock Improved algorithms for linear stochastic bandits.
\newblock In \emph{Advances in Neural Information Processing Systems}, pages
  2312--2320, 2011.

\bibitem[Agrawal et~al.(1989)Agrawal, Teneketzis, and
  Anantharam]{agrawal1989asymptotically}
Rajeev Agrawal, Demosthenis Teneketzis, and Venkatachalam Anantharam.
\newblock Asymptotically efficient adaptive allocation schemes for controlled
  iid processes: Finite parameter space.
\newblock \emph{IEEE Transactions on Automatic Control}, 34\penalty0 (3), 1989.

\bibitem[Burnetas and Katehakis(1997)]{burnetas1997optimal}
Apostolos~N. Burnetas and Michael~N. Katehakis.
\newblock Optimal adaptive policies for \textsc{M}arkov decision processes.
\newblock \emph{Mathematics of Operations Research}, 22\penalty0 (1):\penalty0
  222--255, 1997.

\bibitem[Capp\'{e} et~al.(2013)Capp\'{e}, Garivier, Maillard, Munos, and
  Stoltz]{CaGaMaMuSt2013}
Olivier Capp\'{e}, Aur\'{e}lien Garivier, Odalric-Ambrym Maillard, R\'{e}mi
  Munos, and Gilles Stoltz.
\newblock Kullback--{L}eibler upper confidence bounds for optimal sequential
  allocation.
\newblock \emph{Annals of Statistics}, 41\penalty0 (3):\penalty0 1516--1541,
  2013.

\bibitem[Combes and Proutiere(2014)]{combes2014unimodal}
Richard Combes and Alexandre Proutiere.
\newblock Unimodal bandits: Regret lower bounds and optimal algorithms.
\newblock In \emph{International Conference on Machine Learning}, 2014.

\bibitem[Combes et~al.(2017)Combes, Magureanu, and
  Proutiere]{combes2017minimal}
Richard Combes, Stefan Magureanu, and Alexandre Proutiere.
\newblock Minimal exploration in structured stochastic bandits.
\newblock In \emph{Advances in Neural Information Processing Systems}, pages
  1763--1771, 2017.

\bibitem[Durand et~al.(2017)Durand, Maillard, and Pineau]{durand2017streaming}
Audrey Durand, Odalric-Ambrym Maillard, and Joelle Pineau.
\newblock Streaming kernel regression with provably adaptive mean, variance,
  and regularization.
\newblock \emph{arXiv preprint arXiv:1708.00768}, 2017.

\bibitem[Gentile et~al.(2014)Gentile, Li, and Zappella]{gentile2014online}
Claudio Gentile, Shuai Li, and Giovanni Zappella.
\newblock Online clustering of bandits.
\newblock In \emph{International Conference on Machine Learning}, pages
  757--765, 2014.

\bibitem[Graves and Lai(1997)]{graves1997asymptotically}
Todd~L Graves and Tze~Leung Lai.
\newblock Asymptotically efficient adaptive choice of control laws incontrolled
  markov chains.
\newblock \emph{SIAM journal on control and optimization}, 35\penalty0
  (3):\penalty0 715--743, 1997.

\bibitem[Honda and Takemura(2011)]{honda2011asymptotically}
Junya Honda and Akimichi Takemura.
\newblock An asymptotically optimal policy for finite support models in the
  multiarmed bandit problem.
\newblock \emph{Machine Learning}, 85\penalty0 (3):\penalty0 361--391, 2011.

\bibitem[Honda and Takemura(2015)]{honda2015imed}
Junya Honda and Akimichi Takemura.
\newblock Non-asymptotic analysis of a new bandit algorithm for semi-bounded
  rewards.
\newblock \emph{Machine Learning}, 16:\penalty0 3721--3756, 2015.

\bibitem[Lai(1987)]{lai1987adaptive}
Tze~Leung Lai.
\newblock Adaptive treatment allocation and the multi-armed bandit problem.
\newblock \emph{The Annals of Statistics}, pages 1091--1114, 1987.

\bibitem[Lai and Robbins(1985)]{lai1985asymptotically}
Tze~Leung Lai and Herbert Robbins.
\newblock Asymptotically efficient adaptive allocation rules.
\newblock \emph{Advances in applied mathematics}, 6\penalty0 (1):\penalty0
  4--22, 1985.

\bibitem[Lattimore and Szepesvari(2017)]{lattimore2017end}
Tor Lattimore and Csaba Szepesvari.
\newblock The end of optimism? an asymptotic analysis of finite-armed linear
  bandits.
\newblock In \emph{Artificial Intelligence and Statistics}, pages 728--737,
  2017.

\bibitem[Magureanu(2018)]{magureanu2018efficient}
Stefan Magureanu.
\newblock \emph{Efficient Online Learning under Bandit Feedback}.
\newblock PhD thesis, KTH Royal Institute of Technology, 2018.

\bibitem[Magureanu et~al.(2014)Magureanu, Combes, and
  Proutiere]{magureanu2014oslb}
Stefan Magureanu, Richard Combes, and Alexandre Proutiere.
\newblock Lipschitz bandits: Regret lower bounds and optimal algorithms.
\newblock \emph{Machine Learning}, 35:\penalty0 1--25, 2014.

\bibitem[Maillard(2018)]{maillard2018boundary}
O-A Maillard.
\newblock Boundary crossing probabilities for general exponential families.
\newblock \emph{Mathematical Methods of Statistics}, 27\penalty0 (1):\penalty0
  1--31, 2018.

\bibitem[Maillard and Mannor(2014)]{maillard2014latent}
Odalric-Ambrym Maillard and Shie Mannor.
\newblock Latent bandits.
\newblock In \emph{International Conference on Machine Learning}, pages
  136--144, 2014.

\bibitem[Robbins(1952)]{ro52}
H.~Robbins.
\newblock Some aspects of the sequential design of experiments.
\newblock \emph{Bulletin of the American Mathematics Society}, 58:\penalty0
  527--535, 1952.

\bibitem[Srinivas et~al.(2010)Srinivas, Krause, Kakade, and
  Seeger]{srinivas2010gaussian}
Niranjan Srinivas, Andreas Krause, Sham Kakade, and Matthias Seeger.
\newblock Gaussian process optimization in the bandit setting: no regret and
  experimental design.
\newblock In \emph{Proceedings of the 27th International Conference on
  International Conference on Machine Learning}, pages 1015--1022. Omnipress,
  2010.

\bibitem[Thompson(1933)]{thompson1933likelihood}
William~R Thompson.
\newblock On the likelihood that one unknown probability exceeds another in
  view of the evidence of two samples.
\newblock \emph{Biometrika}, 25\penalty0 (3/4):\penalty0 285--294, 1933.

\bibitem[Thompson(1935)]{thompson1935criterion}
William~R Thompson.
\newblock On a criterion for the rejection of observations and the distribution
  of the ratio of deviation to sample standard deviation.
\newblock \emph{The Annals of Mathematical Statistics}, 6\penalty0
  (4):\penalty0 214--219, 1935.

\bibitem[Valko et~al.(2014)Valko, Munos, Kveton, and
  Koc{\'a}k]{valko2014spectral}
Michal Valko, R{\'e}mi Munos, Branislav Kveton, and Tom{\'a}{\v{s}} Koc{\'a}k.
\newblock Spectral bandits for smooth graph functions.
\newblock In \emph{International Conference on Machine Learning}, 2014.

\end{thebibliography}

\newpage
\appendix

\tableofcontents

\section{Notations and reminders}\label{notations}
For $ \nu \in \cD_\nu$, we define
\[
\epsilon_\nu \coloneqq \min\limits_{\substack{(a,b)\notin \cO^\star, b'\in \cB \\ \mu_{a,b'}- \mu_b^\star - \omegabb \neq 0 }}\Set{ \dfrac{\abs{\mu_{a,b'}- \mu_b^\star - \omegabb}}{4}, \dfrac{\muab}{4}, \dfrac{1-\mu_b^\star}{4} } \,.
\]
Then, there exists $ \alpha_\nu : \Real_+^\star \to \Real_+^\star $ such that $ \lim\limits_{\epsilon\to0}\alpha_\nu(\epsilon) = 0 $ and such that for all $ 0 < \epsilon < \epsilon_\nu $, for all $(a,b) \notin \cO^\star$, for all $b' \in \cB$,
{\small 
\[
\mu_{a,b'} \neq \mu_b^\star - \omegabb   \quad \Rightarrow \quad  \dfrac{\klabb}{1+\alpha_\nu(\epsilon)}\!\leq\! \kl(\mu_{a,b'} \pm \epsilon |\mu_b^\star - \omegabb \pm \epsilon )\!\leq\! (1+\alpha_\nu(\epsilon))\klabb  \,.
\]
}
~\\We also introduce the following constant of interest
$$E_\nu \coloneqq 6\e\,\maxab \left(1- \frac{\log(1-\muab - \epsilon_\nu)}{\log(1-\muab)}\right)^{-1}\left(1-\e^{-\!\left(1- \frac{\log(1-\muab - \epsilon_\nu)}{\log(1-\muab)}\right)\kl(\muab|\muab - \epsilon_\nu)}\right)^{-3} \,.$$
Lastly, for all couple $(a,b) \in \cA\times\cB$, for all $n \geq 1$, we consider the stopping times $$\tau_{a,b}^n \coloneqq \inf{ \Set{t \!\geq\! 1\!: N_{a,b}(t)\! =\! n} } $$ 
and define $$ \muhatab^n \!\coloneqq\! \muhatab(\tau_{a,b}^n)  \,.$$

\section{Proof related to the regret lower bound (Section~\ref{sec:lower_bounds})}
\label{app:proof_lb}
In this section we regroup the proofs related to the lower bounds.

\subsection{Almost sure asymptotic lower bounds  under consistent strategies }
In this section we prove Proposition~\ref{prop:LB_pull}. 

~\\ Let us consider a  consistent strategy on $\cD_\omega$. Let $ \nu  \in \cD_\omega $ and let us consider  $(a,b) \notin \cO^\star $. We show that almost surely $\lim\limits_{T \to \infty}N_b(T) = +\infty$ implies
$$  \liminfT \dfrac{1}{\log\!\left(N_b(T)\right)}  \sum\limits_{b' \in \Bab} N_{a,b'}(T) \kl(\mu_{a,b'}|\mu_b^\star - \omegabb) \geq 1 \,$$
where $ \Bab = \Set{b' \in \cB :\ (a,b') \notin \cO^\star \ad \mu_{a,b'} < \mu_b^\star - \omegabb} $.

\begin{proof}  Let us consider $ \nutild \in \cD_\omega$, a maximal confusing distribution for the sub-optimal couple $(a,b)$ such that $a$ is the unique optimal arm of user $b$ for $\widetilde \nu$, defined as follows: \begin{itemize}[label= -, itemsep = 0 cm]
\item $\forall a' \neq a,\ b' \in \cB, \quad \widetilde{\mu}_{a',b'} = \mu_{a',b'} $
\item $\forall  b'' \notin \cB_{a,b}, \quad \widetilde{\mu}_{a,b''} = \mu_{a,b''} $
\item $\forall  b' \in \cB_{a,b}, \quad \widetilde{\mu}_{a,b'} = \mu_b^\star - \omegabb + \epsilon $
\end{itemize}
where $0 < \epsilon < \epsilon_0 = \min\limits_{b' \in \cB } \abs{\mu_b^\star - \omegabb - \mu_{a,b'}} $. Our assumption on $\cD_\omega\subset\overline{\cD}_\omega $ ensures that $\epsilon_0 > 0$. Note that $\epsilon$ is chosen in such a way that for all $ b',b'' \in \cB$,\ $ \maxa\abs{\mutild_{a,b'} - \mutild_{a,b''}} \leq \omega_{b',b''} $. Indeed we have:
~\\- for $b',b'' \in \Bab : \quad \abs{\mutild_{a,b'} - \mutild_{a,b''}} = \abs{\omega_{b,b'} - \omega_{b,b''}} \leq \omega_{b',b''} $
~\\- for $ b', b'' \notin \Bab : \quad \abs{\mutild_{a,b'} - \mutild_{a,b''}} = \abs{\mu_{a,b'} - \mu_{a,b''}}\leq \omega_{b',b''}   $
~\\ - for $b' \in \Bab$ and $b'' \notin \Bab$, we have $ \mutild_{a,b'} - \mutild_{a,b''} = \mu_b^\star - \omegabb + \epsilon - \mu_{a,b''} $. Since in this case $ b' \in \Bab$ it implies $\mu_{a,b'} \leq  \mu_b^\star - \omega_{b,b'} $ and since $ b''\notin \Bab :  \mu_{a,b''} \geq \mu_b^\star - \omega_{b,b''} + \epsilon_0 $. Therefore on one hand we get
$$ \mu_b^\star - \omegabb + \epsilon - \mu_{a,b''} \geq \mu_{a,b'}  - \mu_{a,b''} \geq - \omega_{b',b''} \,,$$
and on the other hand
$$ \mu_b^\star - \omegabb + \epsilon - \mu_{a,b''} \leq \mu_b^\star - \omegabb + \epsilon - (\mu_b^\star - \omega_{b,b''} + \epsilon_0) = \epsilon - \epsilon_0 + \omega_{b,b''} - \omega_{b,b'} \leq \omega_{b',b''}  \,.$$

~\\Actually, we can choose $ 0 < \epsilon < \epsilon_\nu $ so that :
$$ \forall b' \in \Bab,\ \kl(\mu_{a,b'}|\widetilde{\mu}_{a,b'}) \leq (1+\alpha_\nu(\epsilon))\kl(\mu_{a,b'}|\mu_b^\star - \omegabb)  .$$
We refer to Appendix~\ref{notations} for the definitions of $\epsilon_\nu $ and $\alpha_\nu(\cdot)$. Note that  $\alpha_\nu(\cdot)$ is such that $ \lim\limits_{\epsilon \to 0} \alpha_\nu(\epsilon) = 0 $.

~\\Let $ 0 < c < 1 $ .We will show that almost surely $\lim\limits_{T \to \infty}N_b(T) = + \infty$ implies
$$    \liminfT \dfrac{1}{\log\!\left(N_b(T)\right)}  \sum\limits_{b' \in\Bab} N_{a,b'}(T) \kl(\mu_{a,b'}|\mutild_{a,b'}) \geq c \,.$$
We start with the following inequality
\beqan
\lefteqn{\Pr_\nu\!\left(   \liminfT \dfrac{1}{\log\!\left(N_b(T)\right)}  \sum\limits_{b' \in\Bab} N_{a,b'}(T) \kl(\mu_{a,b'}|\mutild_{a,b'}) < c,\ \limT N_b(T) = \infty \right)}\\
&&\leq \liminfT \Pr_\nu\!\left(\dfrac{1}{\log\!\left(N_b(T)\right)}  \sum\limits_{b' \in\Bab} N_{a,b'}(T) \kl(\mu_{a,b'}|\mutild_{a,b'}) < c,\ \limT N_b(T) = \infty \right)\,. 
\eeqan
Let us consider an horizon $ T \geq 1$ and let us introduce the event
$$ \Omega_T = \Set{ \sum\limits_{ b' \in \Bab} N_{a,b'}(T) \kl(\mu_{a,b'}| \mutild_{a,b'}) < c \log\!\left(N_b(T)\right),\ \limT N_b(T) = \infty} .$$

~\\We want to provide an upper bound on $ \Pr_\nu(\Omega_T) $ to ensure $ \limT\!\Pr_\nu(\Omega_T) = 0 $. We start by taking advantage of the following lemma.
\begin{lemma}[Change of measure]For every  event $\Omega$ and random variable $Z$ both measurable with respect to $\nu$ and $\nutild$,
\[
\Pr_\nu(\Omega\cap E) = \Esp_{\nutild}\!\left[ \dfrac{\textnormal{ d}\nu}{\textnormal{ d}\nutild}(\psi) \ind_{\Set{\Omega\cap E }}\right] \leq \Esp_{\nutild}\!\left[ e^{Z}\ind_{\Set{\Omega }}\right] 
\]
where  $ E = \Set{ \log\left (\dfrac{\textnormal{ d}\nu}{\textnormal{ d}\nutild}(\psi)\right) \leq Z } $ and  $\psi=((a_t,b_t),X_t)_{t = 1 .. T} $ is the sequence of pulled couples and rewards.

\end{lemma}

~\\Let $ \alpha \in (0,1) $ and let us introduce the event
$$ E_T = \Set{ \log\left (\dfrac{\textnormal{ d}\nu}{\textnormal{ d}\nutild}(\psi)\right) \leq (1-\alpha) \log\!\left(N_b(T)\right)} .$$
Then we can decompose the probability as follows 
$$ \Pr_\nu(\Omega_T) = \Pr_\nu(\Omega_T\cap E_T) + \Pr_\nu(\Omega_T\cap E_T^c)  \leq \Esp_{\widetilde \nu}\!\left[ N_b(T)^{1-\alpha}\ind_{\Set{\Omega_T }}\right] + \Pr_\nu(\Omega_T\cap E_T^c) $$
Now, we control successively $ \Esp_{\widetilde \nu}\!\left[ N_b(T)^{1-\alpha}\ind_{\Set{\Omega_T }}\right] $  and $ \Pr_\nu(\Omega_T\cap E_T^c)  $ and show that they both tend to $0$ as $T $ tends to $\infty$.

\subsubsection{ \texorpdfstring{$\Esp_{\widetilde \nu}\!\left[ N_b(T)^{1-\alpha}\ind_{\Set{\Omega_T }}\right]   $}{TEXT} tends to \texorpdfstring{$0$}{TEXT} when \texorpdfstring{$T$}{TEXT} tends to infinity}
We first provide an upper bound on $\ind_{\Set{\Omega_T}} $ as follows, denoting $c' = c/\kl(\mu_{a,b'}|\mutild_{a,b'})$,
\beqan
\Omega_T 
&\subset& \Set{N_{a,b}(T) < c' \log\!\left(N_b(T)\right),\ \limT N_b(T) = \infty} \\
&=& \Set{N_{b}(T) < c' \log\!\left(N_b(T)\right) + \sum\limits_{a' \neq a} N_{a',b}(T),\ \limT N_b(T) = \infty}  \,.
\eeqan
Thus, we have
\[
\ind_{\Set{\Omega_T }} \leq c'\ind_{\Set{N_b(T) \geq 1,\ N_b(T) \to \infty}} \dfrac{\log\!\left(N_b(T)\right)}{N_b(T)} + \sum\limits_{a'\neq a } \dfrac{N_{a',b}(T)}{N_b(T)}
\]
Considering $f_\alpha\!: x \!\geq\! 1 \!\mapsto\! \log(x)/x^\alpha$, we have $f_\alpha \!\leq\! \e^{-1}/\alpha$. Then,  the dominated convergence theorem ensures
\[ 
\Esp_{\widetilde \nu}\!\left[ \ind_{\Set{N_b(T) \geq 1,\ N_b(T) \to \infty}} \dfrac{\log\!\left(N_b(T)\right)}{N_b(T)^\alpha} \right] = o(1) \,.
\]
Furthermore, since the considered strategy is assumed consistent we know that for $a' \neq a$, since $a'$ is a sub-optimal arm for user $b$ and configuration $\widetilde \nu$, 
$$ \Esp_{\nutild}\!\left[\dfrac{N_{a,b'}(T)}{N_b(T)^\alpha}\right] = o(1) \,,$$
therefore we get 
$$ \Esp_{\widetilde \nu}\!\left[ N_b(T)^{1-\alpha}\ind_{\Set{\Omega_T }}\right] = o(1) .$$

\subsubsection{\texorpdfstring{$\Pr_{\nu}(\Omega_T\cap E_T^c)$}{TEXT} tends to \texorpdfstring{$0$}{TEXT} when \texorpdfstring{$T$}{TEXT} tends to infinity}
For each time $t = 1,\ldots,T$, the reward $ X_t  $ is sampled independently from the past and according to $ \nu_{a_t,b_t} $. Hence the likelihood ratio rewrites
$$ \dfrac{\textnormal{ d}\nu}{\textnormal{ d}\nutild}(\psi) = \prod_{t = 1}^T\dfrac{\textnormal{ d}\nu_{a_t,b_t}}{\textnormal{ d}\nutild_{a_t,b_t}}(X_t)  $$
where, for all $(a,b) \in \cA\times\cB$ and for all $ x \in \Set{0,1} $, we have : $  \dfrac{\textnormal{ d}\nu_{a,b}}{\textnormal{ d}\nutild_{a,b}}(x) = \dfrac{\muab^x (1-\muab)^{1-x}}{\mutild_{a,b}^x (1-\mutild_{a,b})^{1-x}}  $ .
~\\ Thus, since for all $ b' \notin \Bab,\ \muab = \mutild_{a,b}$, the log-likelihood ratio is

$$ \log\left (\dfrac{\textnormal{ d}\nu}{\textnormal{ d}\nutild}(\psi)\right) = \sum\limits_{b' \in \Bab}\sum_{t=1}^T \ind_{\Set{ (a_t,b_t) = (a,b') }} \log\left(\dfrac{\textnormal{ d}\nu_{a,b'}}{\textnormal{ d}\nutild_{a,b'}}(X_t)\right) \,.$$
~\\
Let us introduce, for $(a,b) \in \cA\times\cB$,  $ X_{a,b}^n =  X_{ \tau_{a,b}^n } $ where $ \tau_{a,b}^n = \min\Set{t \geq 1 \st N_{a,b}(t) = n }$. Note that the random variables $ \tau_{a,b}^n$ are predictable stopping times, since $ \Set{ \tau_{a,b}^n  = t } $ is measurable with respect to the filtration generated by $ ((a_1,b_1), X_1, ..., (a_{t-1},b_{t-1}), X_{t-1} ) $.  
Hence we can rewrite the event $E_T$
$$ E_T = \Set{ \sum\limits_{b' \in \Bab}\sum\limits_{t = 1}^T \ind_{\Set{(a_t,b_t) =(a,b)}} \dfrac{\textnormal{ d}\nu_{a,b'}}{\textnormal{ d}\nutild_{a,b'}}(X_t) \leq (1-\alpha) \log\!\left(N_b(T)\right) }   $$
and, since  $ \Omega_T = \Set{ \sum\limits_{ b' \in \Bab} N_{a,b'}(T) \kl(\mu_{a,b'}| \mutild_{a,b'}) < c \log\!\left(N_b(T)\right),\ \limT N_b(T) = \infty}$, we have 
\beqan
\Omega_T\cap E_T^c &\subset  \bigg\{ \exists (n_{b'})_{b' \in \Bab}: \ \sum\limits_{ b' \in \Bab } n_{b'} \kl(\mu_{a,b'}| \mutild_{a,b'}) < c \log\!\left(N_b(T)\right),\ \limT N_b(T) = \infty&\\
&\textnormal{ and } \sum\limits_{b' \in \Bab}\sum\limits_{n = 1 .. n_{b'}}  \dfrac{\textnormal{ d}\nu_{a,b'}}{\textnormal{ d}\nutild_{a,b'}}(X_{a,b'}^n) > (1-\alpha) \log\!\left(N_b(T)\right) &\bigg\}\,.
\eeqan
For $ b' \in \Bab$ and $ n \geq 1$, let us consider $ Z_{b'}^n = \dfrac{\textnormal{ d}\nu_{a,b'}}{\textnormal{ d}\nutild_{a,b'}}(X_{a,b'}^n) $. Then $ Z_{b'}^n$ is positive and bounded by $B_{b'} =  \dfrac{1}{\mutild_{a,b'}(1 - \mutild_{a,b'})} $, with mean $ \Esp_\nu[Z_{b'}^n] = \kl(\mu_{a,b'}|\mutild_{a,b'}) $ . Furthermore, the  random variables $Z_{b'}^n $, for $b'\in \Bab$ and $ n \geq 1$, are independent.
Thus, it holds
$$ \Omega_T\cap E_T^c \subset \Set{ \max\limits_{(n_b') \in \cN_{a,b}} \sum\limits_{b' \in \Bab}\sum\limits_{n = 1 .. n_{b'}} \!\!\!\!\! Z_{b'}^n - \Esp_\nu[Z_{b'}^n]  > \left(\dfrac{1-\alpha}{c} - 1 \right)c \log\!\left(N_b(T)\right) ,\  N_b(T) \!\to\! \infty \!}, $$
where 
$$ \cN_{a,b} \coloneqq \Set{(n_{b'})_{b' \in \Bab}:\ \sum\limits_{b' \in \Bab} n_{b'} \kl(\mu_{a,b'}|\mutild_{a,b'}) < c \log(T)  }  .  $$
In the following, we apply Doob's maximal inequality. For $ b' \in \Bab$ and $\lambda > 0 $, let us introduce the super-martingale
$$ M_n^{b'} = \exp\!\left( \lambda \sum\limits_{k = 1}^n  \big(Z_{b'}^k - E [Z_{b'}^k]\big) - n \lambda^2 \dfrac{B_{b'}^2}{8} \right).$$
Then noting that 
$$ \dfrac{\sum\limits_{ b' \in \Bab } \lambda^2 n_{b'}\frac{B_{b'}^2}{8}}{c\log(T)}   < \dfrac{\sum\limits_{ b' \in \Bab } \lambda^2 n_{b'}\frac{B_{b'}^2}{8}}{\sum\limits_{ b' \in \Bab } n_{b'} \kl(\mu_{a,b'}| \mutild_{a,b'} )}   \leq \lambda^2 \dfrac{ \max\limits_{b'\in \Bab}B_{b'}^2}{8\min\limits_{b'\in \Bab}\kl(\mu_{a,b'}| \mutild_{a,b'} )}\,,  $$
we obtain  
\beqan  \Omega_T\cap E_T^c &\subset& \Set{ \max\limits_{(n_b') \in \cN_{a,b}} \prod\limits_{b' \in \Bab} M_{n_b'}^{b'} > T^{\Big[ \lambda (\frac{1 - \alpha}{c} - 1 ) - \lambda^2 \frac{\max_{ b' \in \Bab } B_{b'}^2}{8\min_{b'\in \Bab}\kl(\mu_{a,b'}| \mutild_{a,b'} )}\Big] c } ,\  N_b(T)\!\to\!\infty }  \\
&\subset& \Set{\exists b' \in \Bab:\  \max\limits_{  n \leq  n_{\max}  } M_n^{b'} > N_b(T)^\gamma,\ N_b(T) \!\to\!\infty }\,,
\eeqan
where $ n_{\max} = \frac{c \log(T)}{\min\limits_{b'\in \Bab}\kl(\mu_{a,b'}| \mutild_{a,b'} )} $ and $ \gamma= \Big[ \lambda (\frac{1 - \alpha}{c} - 1 ) - \lambda^2 \frac{\max_{ b' \in \Bab } B_{b'}^2}{8\min_{b'\in \Bab}\kl(\mu_{a,b'}| \mutild_{a,b'} )}\Big] \dfrac{c}{\abs{\Bab}}   $. In order to have $\gamma > 0$, we impose:
\\ - $ 0 < \alpha < 1-c $ \quad (this implies $\frac{1-\alpha}{c} - 1 > 0$)
~\\ - $ \lambda \in \argmax\limits_{\lambda' \geq 0}\Set{ \lambda' (\dfrac{1 - \alpha}{c} - 1 ) - {\lambda'}^2 \frac{\max_{ b' \in \Bab } B_{b'}^2}{8\min_{b'\in \Bab}\kl(\mu_{a,b'}| \mutild_{a,b'} )} } > 0 $.
~\\Thus for $A > 0 $, we have
\beqan
\Pr_\nu(\Omega_T\cap E_T^c) 
&\leq& \sum\limits_{b' \in \Bab}\Pr_\nu\!\left(\max\limits_{  n \leq  n_{\max}  } M_n^{b'} > N_b(T)^\gamma ,\ N_b(T) \!\to\!\infty \right) \quad \textnormal{(Union bound)} \\
&\leq& \sum\limits_{b' \in \Bab}\Pr_\nu\!\left( N_b(T)^\gamma \leq A,\ N_b(T) \!\to\!\infty\right) + \Pr_\nu\!\left(\max\limits_{  n \leq  n_{\max}  } M_n^{b'} > A \right)  \\
& \leq & \abs{\Bab} \Pr_\nu\!\left( N_b(T)^\gamma \leq A,\ N_b(T) \!\to\!\infty\right) + \sum\limits_{b' \in \Bab} \dfrac{\Esp_\nu[M_0^{b'}]}{A} \quad \textnormal{(Doob's maximal inequality)} \\
&=& \abs{\Bab} \Pr_\nu\!\left( N_b(T)^\gamma \leq A,\ N_b(T) \!\to\!\infty\right) + \dfrac{\abs{\Bab}}{A} \,.
\eeqan
Furthermore, we have
\beqan
\limT \Pr_{\widetilde \nu}\!\left( N_b(T)^\gamma 
\leq A,\ N_b(T) \!\to\!\infty \right) 
&\leq&  \Pr_{\widetilde \nu}\!\left(\limsupT \left(N_b(T)^\gamma \leq A\right),\ N_b(T) \!\to\!\infty \right) \\
&\leq& \Pr_{\widetilde \nu}\!\left( \limsupT N_b(T) < \infty,\ N_b(T) \!\to\!\infty \right) = 0 \,.
\eeqan
Thus we have shown
\[
\forall A > 0,\quad \limsupT\Pr_\nu(\Omega_T\cap E_T^c) \leq \dfrac{\abs{\Bab}}{A} \,,
\]
that is $\Pr_\nu(\Omega_T\cap E_T^c) = o(1)$.

\end{proof}

\subsection{Asymptotic lower bounds on the regret}
Here, we explain how we obtain the lower bounds on the regret given in Corollary~\ref{cor:LB_regret}.

\begin{proof}[Proof of Corollary~\ref{cor:LB_regret}.] Let us consider a  consistent strategy on $\cD_\omega$ and let $\nu \in \cD_\omega$. Let $(T_k)_{k\in\Nat}$ be a sub-sequence such that 
\[
\liminf_{T\to\infty}\frac{R(T,\nu)}{\log(T)}= \lim_{k\to \infty} \frac{R(T_k,\nu)}{\log(T_k)}\,.
\]
We assume that this limit is finite otherwise the result is straightforward. This implies in particular that for all $(a,b)\notin \cO^\star$
\[
\limsup_{k\to\infty} \frac{\Esp_\nu\left[N_{a,b}(T_k)\right]}{\log(T_k)}<+\infty\,.
\]
By Cantor's diagonal argument there exists an extraction of $(T_k)_{k\in\Nat}$ denoted by $(T_k')_{k\in\Nat}$ such that for all $(a,b)\notin \cO^\star$, there exist $N_{a,b}\not \in \cO^\star$ such that 
\[
\lim_{k'\to\infty}  \frac{\Esp_\nu\left[N_{a,b}(T_{k}')\right]}{\log(T_{k}')} = N_{a,b}\,.
\]
Hence we get 
\[
\liminf_{T\to\infty}\frac{R(T,\nu)}{\log(T)} = \sum_{(a,b)\notin \cO^\star} N_{a,b}\Delta_{a,b}\,.
\]
But thanks to Proposition~\ref{prop:LB_pull} we have for all $(a,b)\not\in \cO^\star$, since user $b$ has a log-frequency $\beta_b$,
\begin{align*}
&\sum\limits_{b'\in \Bab}\kl(\mu_{a,b'}|\mu_{b}^\star - \omega_{b,b'} ) N_{a,b'} \\
&= \lim_{k\to \infty} \sum\limits_{b'\in \Bab}\kl(\mu_{a,b'}|\mu_{b}^\star - \omega_{b,b'} )  \frac{\Esp_\nu \big[N_{a,b'}(T_{k}')\big]}{\log(T_{k}')}\\
&\geq  \liminf_{k\to\infty}\sum\limits_{b'\in \Bab}\kl(\mu_{a,b'}|\mu_{b}^\star - \omega_{b,b'} ) \frac{\Esp_\nu \big[N_{a,b'}(T_{k}')\big]}{\log\!\left(N_b(T'_k)\right)} \times \liminf_{k \to \infty}\dfrac{\log\!\left(N_b(T'_k)\right)}{\log(T'_k)}\geq \beta_b \,. \\
\end{align*}
Therefore we obtain the lower bound 
\beqan
\liminfT \dfrac{R(\nu,T)}{\log(T)} \geq C_\omega^\star(\beta,\nu):=& \min\limits_{n \in \Real_+^{\cA\times\cB} }&  \sum\limits_{a,b \notin \cO^\star} n_{a,b} \deltaab   
\\
& s.t. & \forall (a,b) \not\in \cO^\star: \quad \sum\limits_{b' \in \Bab}\kl(\mu_{a,b'}|\mu_b^\star - \omegabb) n_{a,b'} \geq \beta_b\,.    
\eeqan

\end{proof}
\section{Algorithms for the \textit{uncontrolled scenario}}
\label{app:algo_senario2}
We regroup in this section the algorithms \IMEDStwo  and \IMEDSstartwo  for the \textit{uncontrolled scenario}.

\vspace{-6mm}
\begin{figure}[H]
\centering
\makebox[0pt][c]{%
\begin{minipage}[]{0.35\textwidth}
\small
\begin{algorithm}[H]
\small
\caption{\IMEDStwo}
\label{alg:IMEDS2}
\begin{algorithmic}
\REQUIRE Weight matrix $(\omegabb)_{b,b' \in \cB}  $.
\FOR {$ t = 1 ... T $} 
\STATE Pull  $ a_{t+1} \in \argmin\limits_{a \in \cA }\widetilde I_{a,b_{t+1}}(t) $
\ENDFOR

\vspace{148.5 mm}
\end{algorithmic}
\end{algorithm}
\end{minipage}
\hfill
\begin{minipage}[]{0.6\textwidth}
\begin{algorithm}[H]
\small
\caption{\IMEDSstartwo}
\label{alg:IMEDSstar2}
\begin{algorithmic}
\REQUIRE Weight matrix $(\omegabb)_{b,b' \in \cB}  $.
\STATE $\forall a \in \cA,\ c_a, c_a^+ \leftarrow 1 $
\STATE $\forall b \in \cB,\ \textnormal{E}(b) \leftarrow \emptyset,\ \textnormal{FE}(b) \leftarrow \emptyset$
\FOR{ For $ t = 1 ... T $}
\STATE Choose  $\aol_t \in \argmin\limits_{a \in \cA}\widetilde I_{a,b_{t+1}}(t)$
\IF{$(\aol_t,b_{t+1})\in \Ohat^\star(t)$ }
\STATE Choose $a_{t+1} = \aol_t$
\ELSE
\STATE
\STATE Choose  $ (\aul_t,\bul_t) \in \argmin\limits_{(a,b)\in \cA\times\cB}\widetilde I_{a,b}(t) $ 
\IF{$(\aul_t,\bul_t)\notin \Ohat^\star(t)$}
\IF{$c_{\aul_t} = c_{\aul_t}^+ $}
\STATE $c_{\aul_t}^+ \leftarrow 2 c_{\aul_t}^+   $
\STATE Choose   $\bol_t \in \argminb \Nab(t) $ 
\STATE $\textnormal{FE}\!\left(\bol_t\right) \leftarrow \aul_t$
\ELSE 
\STATE Choose $\bol_t \!\in\!\!\!\!\!\!\!\!\! \argmax\limits_{b \in \Bhat_{{\aul_t},{\bul_t}}(t)\cup\Set{\bul_t}} \Nopt_{a_{t+1},b}(t) \! - \!N_{a_{t+1},b}(t)  $
\STATE $\textnormal{E}\!\left(\bol_t\right) \leftarrow \aul_t$
\ENDIF
\STATE $c_{a_{t+1}} \leftarrow c_{a_{t+1}}  + 1 $
\ENDIF
\STATE
\STATE \textbf{Priority rule in exploration phases}:
\IF{$\textnormal{FE}(b_{t+1})\neq\emptyset$}    
\STATE Choose $a_{t+1} \!=\!\textnormal{FE}(b_{t+1})$\ \textnormal{(delayed forced-exploration)} 
\STATE $\textnormal{FE}(b_{t+1}) \leftarrow \emptyset$
\ELSIF{$\textnormal{E}(b_{t+1})\neq\emptyset$}    
\STATE Choose $a_{t+1} \!=\! \textnormal{E}(b_{t+1})$ \ \textnormal{(delayed exploration)}
\STATE $\textnormal{E}(b_{t+1}) \leftarrow \emptyset$
\ELSE 
\STATE Choose $a_{t+1} \!=\! \aol_t$    \ \textnormal{(current exploration)}
\ENDIF
\STATE
\ENDIF
\STATE
\STATE Pull $a_{t+1}$
\ENDFOR
\end{algorithmic}
\end{algorithm}
\end{minipage}
}%
\end{figure}

\newpage
\section{\label{app: IMEDSstar finite time analysis}\texorpdfstring{\IMEDSstar}{TEXT}: Finite-time analysis}
\IMEDSstar strategy implies empirical lower and empirical upper bounds on the numbers of pulls (Lemma~\ref{imedstar empirical lower bounds}, Lemma~\ref{imedstar empirical upper bounds}). Based on concentration lemmas (see Appendix~\ref{subsec : imed_concentration}), the strategy-based empirical lower bounds ensure the reliability of the estimators of interest (Lemma~\ref{imedstar reliability}). Then, combining the reliability of these estimators with the obtained strategy-base empirical upper bounds, we get upper bounds on the average numbers of pulls (Proposition~\ref{prop:upper bounds}). We first show that \IMEDSstar strategy is Pareto-optimal (for minimization problem~\ref{eq:lb_regret_structure}) and that it is a consistent strategy which induces sequences of users with log-frequencies all equal to $1$ (independently from the considered bandit configuration). From an asymptotic analysis, we then prove that \IMEDSstar strategy is asymptotically optimal.


	\subsection{Strategy-based empirical bounds} \label{imedstar empirical bounds}
	\IMEDSstar strategy implies inequalities between the indexes that can be rewritten as inequalities on the numbers of pulls. While asymptotic analysis suggests lower bounds involving $\log\!\left(N_{b_{t+1}}(t)\right)$ might be expected, we show in this non-asymptotic context lower bounds on the numbers of pulls involving instead the logarithm of the number of pulls of the current chosen arm, $\log\!\left(N_{a_{t+1},b_{t+1}}(t)\right)$. In contrast, we provide upper bounds involving $\log\!\left(N_{b_{t+1}}(t)\right)$ on $N_{a_{t+1},b_{t+1}}(t)$.

	We believe that establishing these empirical lower and upper bounds is a key element of our proof technique, that is of independent interest and not \textit{a priori} restricted to the graph structure.
	\begin{lemma}[Empirical lower bounds]\label{imedstar empirical lower bounds}Under \IMEDSstar, at each step time $t \geq 1$, for all couple $(a,b) \notin \Ohat^\star(t)$,
		\[ 
         \log\!\left(N_{a_{t+1}, b_{t+1}}(t)\right)  \leq \sum\limits_{b'\in \Bhat_{a,b}(t)} N_{a,b'}(t) \, \klof{\muhat_{a,b'}(t)}{\muhat_{b}^\star(t) - \omegabb}   + \log\!\left(N_{a,b'}(t)\right) \,.
        \]
    Furthermore, for all couple $(a,b) \in \Ohat^\star(t)$,
        \[
          N_{a_{t+1},b_{t+1}}(t) \leq N_{a,b}(t) \,.
        \]
	\end{lemma}
	
	\begin{proof} According to \IMEDSstar~strategy (see Algorithm~\ref{alg:IMEDSstar}),  $a_{t+1} = \aul_t$ and for all couple $(a,b) \in \cA\times\cB$ 
	\[
	I_{a,b}(t) \geq I_{\aul_t,\bul_t}(t)  \,.
	\]
	There is three possible cases.
	~\\\underline{{Case\,1}}: $(a_{t+1},b_{t+1}) = (\aul_t,\bul_t) \in \Ohat^\star(t)$ and $I_{\aul_t,\bul_t}(t) = \log\!\left(N_{a_{t+1},b_{t+1}}(t)\right)$.
	~\\\underline{{Case\,2}}: $b_{t+1} \in \Bhat_{\aul_t,\bul_t}(t)\cup\Set{\bul_t}$ and ~\\$I_{\aul_t,\bul_t}(t) = \sum\limits_{b'\in \Bhat_{\aul_t,\bul_t}(t)} N_{\aul_t,b'}(t) \, \klof{\muhat_{\aul_t,b'}(t)}{\muhat_{\bul_t}^\star(t) - \omega_{\bul_t,b'}}  + \log(N_{\aul_t,b'}(t)) $. Note that $\bul_t \in \Bhat_{\aul_t,\bul_t}(t)$ except if $N_{\aul_t,\bul_t}(t) = 0$. Thus $ I_{\aul_t,\bul_t}(t) \geq \log\!\left(N_{a_{t+1},b_{t+1}}(t)\right) $.
	~\\\underline{{Case\,3}}: $b_{t+1} \in \argminb N_{\aul_t,b}(t)$ and $I_{\aul_t,\bul_t}(t) \geq \min_{b \in \Bhat_{\aul_t,\bul_t}(t)}\log\!\left(N_{a_{t+1},b}(t)\right)\geq\log\!\left(N_{a_{t+1},b_{t+1}}(t)\right)$.
	~\\This implies for all couple $(a,b) \in \cA\times\cB$,
	\[
	 I_{a,b}(t) \geq \log\!\left( N_{a_{t+1},b_{t+1}}(t)\right)  \,.
	\]
	Thus, according to the definition of the indexes (Eq.~\ref{eq:def_imed_index_couple}) for all couple $(a,b) \notin \Ohat^\star(t)$ we obtain
	\[
     \log\left(N_{a_{t+1}, b_{t+1}}(t)\right)  \leq \sum\limits_{b'\in \Bhat_{a,b}(t)} N_{a,b'}(t) \, \klof{\muhat_{a,b'}(t)}{\muhat_{b}^\star(t) - \omegabb }   + \log\left(N_{a,b'}(t)\right) \,,
     \]
     and for all couple $(a,b) \in \Ohat^\star(t)$,
     \[
     \log\left(N_{a_{t+1}, b_{t+1}}(t)\right)  \leq \log\left(N_{a,b}(t)\right) \,.
     \]
     Taking the exponential in the last inequality allows us to conclude.
	\end{proof}
	
	\begin{lemma}[Empirical upper bounds]\label{imedstar empirical upper bounds}
		Under \IMEDSstar, at each step time $t \geq 1$ such that ~\\$(a_{t+1},b_{t+1}) \notin \Ohat^\star(t)$ we have
		\[
		\sum\limits_{b'\in \Bhat_{a_{t+1},\bul_t}(t)} N_{a_{t+1},b'}(t)\, \klof{\muhat_{a_{t+1},b'}(t)}{\muhat_{\bul_t}^\star(t) - \omega_{\bul_t,b'}} \leq \log\!\left(N_{\bul_t}(t)\right)
		\]
		and
		\[
		\dfrac{N_{a_{t+1},b_{t+1}}(t)}{\log(t)} \leq \left\{ \begin{array}{l}
		    \dfrac{1}{\klof{\muhat_{a_{t+1},\bul_t}(t)}{\muhat_{\bul_t}^\star(t)}}  \textnormal{, if } c_{a_{t+1}} = c_{a_{t+1}}^+ \\
		    \min\left(\dfrac{1}{\klof{\muhat_{a_{t+1},b_{t+1}}(t)}{\muhat_{\bul_t}^\star(t) - \omega_{\bul_t,b_{t+1}}}},\ \nopt_{a_{t+1},b_{t+1}}(t) \right) \ \textnormal{ otherwise.} 
		\end{array} \right.
		\]
	\end{lemma}
	\begin{proof} For all current optimal couple $(a,b) \in \Ohat^\star(t)$, we have $ I_{a,b}(t) = \log\left(N_{a,b}(t)\right) \leq \log\!\left(N_b(t)\right)$. 
	~\\ This implies
	\[
	I_{\aul_t,\bul_t}(t) \leq \log\!\left(N_{\bul_t}(t)\right) \,.
	\]
	Furthermore, since $(a_{t+1},b_{t+1}) \notin \Ohat^\star(t)$, we have $a_{t+1} = \aul_t$ and the previous inequality implies
	\beq \label{eq:upper_bounds_1}
	\sum\limits_{b'\in \Bhat_{a_{t+1},\bul_t}(t)} N_{a_{t+1},b'}(t)\, \klof{\muhat_{a_{t+1},b'}(t)}{\muhat_{\bul_t}^\star(t) - \omega_{\bul_t,b'}} \leq \log\!\left(N_{\bul_t}(t)\right) \,.
	\eeq
	In the following we study separately the two cases either $c_{a_{t+1}} = c_{a_{t+1}}^+$ or $c_{a_{t+1}} < c_{a_{t+1}}^+$.
	~\\\underline{Case\,1}: $c_{a_{t+1}} = c_{a_{t+1}}^+$
	~\\ Then $b_{t+1} \in \argmin\limits_{b \in \cB} N_{a_{t+1},b}(t)$ and from Eq.~\ref{eq:upper_bounds_1} we get 
	\[
	N_{a_{t+1},b_{t+1}}(t) \leq N_{a_{t+1},\bul_t}(t) \leq  \dfrac{\log(t)}{\klof{\muhat_{a_{t+1},\bul_t}(t)}{\muhat_{\bul_t}^\star(t)}} \, .
	\]
	
	~\\\underline{Case\,2}: $c_{a_{t+1}} < c_{a_{t+1}}^+$
	~\\Then $b_{t+1} \in \argmax\limits_{b \in \Bhat_{\aul_t,\bul_t}(t)\cup\Set{\bul_t}}\Nopt_{\aul_t,b}(t) - N_{\aul_t,b}(t)$ and we have
	\[
	N_{a_{t+1},b_{t+1}}(t) \leq \dfrac{\log(t)}{\klof{\muhat_{a_{t+1},b_{t+1}}(t)}{\muhat_{\bul_t}^\star(t) - \omega_{\bul_t,b_{t+1}}}} \,.
	\]
	Last, Lemma~\ref{nopt dominates} stated below implies
	\[
	N_{a_{t+1},b_{t+1}}(t) \leq \Nopt_{a_{t+1},b_{t+1}}(t) = \nopt_{a_{t+1},b_{t+1}}(t) \, \min\limits_{b \in \cB} I_{a_{t+1},b}(t) \leq \nopt_{a_{t+1},b_{t+1}}(t) \, \log(t) \,.
	\]
	\end{proof}
	\begin{lemma}[$\Nopt$ dominates $N$]\label{nopt dominates} Under \IMEDSstar, at each time step $t \geq 1$ such that ~\\$(\aul_t,\bul_t) \notin \Ohat^\star(t)$ we have 
	\[
	\max\limits_{b \in \Bhat_{\aul_t,\bul_t}(t)\cup\Set{\bul_t}} \Nopt_{\aul_t,b}(t) - N_{\aul_t,b}(t) \geq 0 \,.
	\]
	\end{lemma}
	\begin{proof}If $\bul_t \notin \Bhat_{\aul_t,\bul_t}(t) $, then $N_{\aul_t,\bul_t}(t) = 0$ and $  \Nopt_{\aul_t,\bul_t}(t) - N_{\aul_t,\bul_t}(t) \geq 0$. In the following we assume that $\bul_t \in \Bhat_{\aul_t,\bul_t}(t) \neq \emptyset$.
	
	~\\ From Eq.~\ref{eq:Nopt}, since $\min\limits_{b' \in \cB}I_{\aul_t,b'}(t) = I_{\aul_t,\bul_t}(t)$, for all $b \in \Bhat_{\aul_t,\bul_t}(t)$ we have 
	\beq \label{eq:NoptN_1}
	\Nopt_{\aul_t,b}(t) = \nopt_{\aul_t,b}(t) \, I_{\aul_t,\bul_t}(t)\,,
	\eeq
	and, since $(\aul_t,\bul_t) \notin \Ohat^\star(t)$, from Eq.~\ref{eq:current_optimization_pb} we get 
	\beq \label{eq:NoptN_2}
	\sum\limits_{b' \in \Bhat_{\aul_t,\bul_t}(t)}\klof{\muhat_{\aul_t,b'}(t)}{\muhat_{\bul_t}^\star(t) - \omega_{\bul_t,b'}}\, \nopt_{\aul_t,b'}(t) \geq 1 \,.
	\eeq 
	Then Eq.~\ref{eq:NoptN_1} and \ref{eq:NoptN_2} imply
	\[
	\sum\limits_{b' \in \Bhat_{\aul_t,\bul_t}(t)}\Nopt_{\aul_t,b'}(t) \,\klof{\muhat_{\aul_t,b'}(t)}{\muhat_{\bul_t}^\star(t) - \omega_{\bul_t,b'}} \geq I_{\aul_t,\bul_t}(t) \,.
	\] 
	Hence from the definitions of the indexes (Eq.~\ref{eq:def_imed_index_couple}) this implies
	\begin{align} \label{eq:NoptN_3}
	\sum\limits_{b' \in \Bhat_{\aul_t,\bul_t}(t)} \Nopt_{\aul_t,b'}(t)& \,\klof{\muhat_{\aul_t,b'}(t)}{\muhat_{\bul_t}^\star(t) - \omega_{\bul_t,b'}} \\
	&\geq  \sum\limits_{b' \in \Bhat_{\aul_t,\bul_t}(t)}N_{\aul_t,b'}(t) \,\klof{\muhat_{\aul_t,b'}(t)}{\muhat_{\bul_t}^\star(t) - \omega_{\bul_t,b'}}  \,. \nonumber
	\end{align}
	Since we assume $\bul_t \in \Bhat_{\aul_t,\bul_t}(t) \neq \emptyset$, previous Eq.~\ref{eq:NoptN_3} implies
	\[
	\max\limits_{b \in \Bhat_{\aul_t,\bul_t}(t)\cup\Set{\bul_t}} \Nopt_{\aul_t,b}(t) - N_{\aul_t,b}(t) =\max\limits_{b \in \Bhat_{\aul_t,\bul_t}(t)} \Nopt_{\aul_t,b}(t) - N_{\aul_t,b}(t) \geq 0 \,.
	\]
	\end{proof}
	
\subsection{\label{subsec: reliable estimators}Reliable current best arm and  means}
	In this subsection, we consider the subset $\cT_{\epsilon,\gamma}$ of times where everything behaves well, that is: The current best couples correspond to the true ones, and, the empirical means of the best couples and the couples at least pulled proportionally (with a coefficient $\gamma \in (0,1)$) to  the number of pulls of the current chosen couple are $\epsilon$-accurate for $0<\epsilon<\epsilon_\nu$, i.e.
	\[
		\cT_{\epsilon, \gamma} \coloneqq \Set{t \geq 1:\ \begin{array}{l}
		  \Ohat^\star(t) = \cO^\star  \\  \forall (a,b) \textnormal{ s.t. } N_{a,b}(t) \geq \gamma\, N_{a_{t+1},b_{t+1}}(t) \textnormal{ or } (a,b) \in \cO^\star, \ \abs{\muhatab(t) - \muab } < \epsilon 
		\end{array}   }\,. 
		\]
	 We will show that its complementary is finite on average. In order to prove this we decompose the set $\cT_{\epsilon,\gamma}$ in the following way. Let $\cE_{\epsilon,\gamma}$ be the set of times where the means are well estimated,
	 \[
		\cE_{\epsilon,\gamma} \coloneqq \Set{t \geq 1:\ \forall (a,b) \textnormal{ s.t. } N_{a,b}(t) \geq \gamma\, N_{a_{t+1},b_{t+1}}(t) \textnormal{ or } (a,b) \in \Ohat^\star(t), \ \abs{\muhatab(t) - \muab } < \epsilon  }\,,
	  \]
	 and $\Lambda_\epsilon$ the set of times where the mean of a couple that is not the current optimal neither pulled is underestimated
	 {\small 
			\[
			\Lambda_{\epsilon}\! \coloneqq \!\Set{\!t \geq 1 \left| \exists a\in \cA, \cB'\! \subset\! \cB\! : \! \begin{array}{l}
				\forall b \in \cB', \ \muhatab(t) < \muab - \epsilon\\	\log(N_{a_{t+1},b_{t+1}}(t))\! \leq \!
				\sum\limits_{b \in \cB'}\Nab(t) \klof{\muhatab(t)}{\muab - \epsilon_\nu}\!  +\! \log\left(\Nab(t)\right) 
				\end{array} \right. \!}\!.
			\]
		}
		Then we prove below the following inclusion. 		
\begin{lemma}[Relations between the subsets of times]\label{lem:decomposition_T_epsilon}For $ 0 < \epsilon < \epsilon_\nu$ and $\gamma \in (0,1)$,
\begin{equation}
    \cT_{\epsilon,\gamma}^c\setminus\cE_{\epsilon,\gamma}^c  \subset \Lambda_{\epsilon_\nu}  \,.
\label{eq:decomp_T_epsilon}
\end{equation}
Refer to Appendix~\ref{notations} for the definition of $\epsilon_\nu$.
\end{lemma}
\begin{proof} For all user $b \in \cB$ it is assumed that there exists a unique optimal arm $a_b^\star \in \cA$ such that $(a_b^\star,b) \in \cO^\star$. We have $\cO^\star = \bigcup\limits_{b \in \cB}\Set{(a_b^\star,b)}$. In particular, for all time step $t \geq 1$, if $\Ohat^\star(t) \neq \cO^\star$ then there exists $b \in \cB$ and $\ahat_b^\star \in \cA$ such that $(\ahat_b^\star,b) \in \Ohat^\star(t) \setminus \cO^\star$ (and $\ahat_b^\star \neq a_b^\star$).

~\\Let $ t  \in \cT_{\epsilon,\gamma}^c\setminus\cE_{\epsilon,\gamma}^c$.  Then $\Ohat^\star(t) \neq \cO^\star$ and there exists $b \in \cB,\, \ahat_b^\star \in \cA$ such that $(\ahat_b^\star,b) \in \Ohat^\star(t) \setminus \cO^\star$.
Thus we know that $ \ahat^\star_b \neq a_b^\star$. In particular, we have  $\mu_{a_b^\star,b} = \mu_b^\star > \mu_{\ahat^\star_b,b} + 2 \epsilon_\nu $. Since $t \in \cE_{\epsilon,\gamma}$ and $\epsilon < \epsilon_\nu$ , this implies 
\beq \label{eq:inclusion_1}
\mu_{a_b^\star,b} > \muhat_{\ahat_b^\star,b}(t) + \epsilon_\nu = \muhat_b^\star(t) + \epsilon_\nu \geq \muhat_{a_b^\star}(t) + \epsilon  
\eeq
and $(a_b^\star,b) \notin \Ohat^\star(t)$. 
From Lemma~\ref{imedstar  empirical lower bounds} we have the following empirical lower bound
\beq \label{eq:inclusion_2}
\log\left(N_{a_{t+1}, b_{t+1}}(t)\right)  \leq \sum\limits_{b'\in \Bhat_{a_b^\star,b}(t)} N_{a_b^\star,b'}(t) \, \klof{\muhat_{a_b^\star,b'}(t)}{\muhat_{b}^\star(t) - \omegabb}   + \log\left(N_{a_b^\star,b'}(t)\right) \,.
\eeq
In particular, for all $b' \in \Bhat_{a_b^\star,b}(t)$  we have $\muhat_{a_b^\star,b'}(t)  < \muhat^\star_b(t) - \omegabb$ and Eq.~\ref{eq:inclusion_1} implies
\beq \label{eq:inclusion_3}
\muhat_{a_b^\star,b'}(t)  < \muhat^\star_b(t) - \omegabb <  \mu_{a_b^\star,b} - \epsilon_\nu - \omegabb < \mu_{a_b^\star,b'} - \epsilon_\nu \,,  
\eeq
and the monotonic properties of $\kl(\cdot|\cdot)$ implies 
\beq \label{eq:inclusion_4}
\klof{\muhat_{a_b^\star,b'}(t)}{\muhat_b^\star(t) - \omegabb} \leq \klof{\muhat_{a_b^\star,b'}(t)}{\mu_{a_b^\star,b'} - \epsilon_\nu} \,.
\eeq
Therefore, by combining Eq. \ref{eq:inclusion_2}, \ref{eq:inclusion_3} and \ref{eq:inclusion_4},  we have for such $t$
\[
\forall b' \in \Bhat_{a_b^\star,b}(t),\ \muhat_{a_b^\star,b'}(t)  < \mu_{a_b^\star,b'} - \epsilon_\nu
\]
and
\[
\log\left(N_{a_{t+1},b_{t+1}}(t)\right)  \leq \sum_{b'\in \Bhat_{a^\star_b,b}(t)}N_{a_b^\star,b'}(t) \, \klof{\muhat_{a_b^\star,b'}(t)}{\mu_{a_b^\star,b'} - \epsilon_\nu} + \log\left(N_{a_b^\star,b'}(t)\right) \,,
\]
which concludes the proof.
\end{proof}
Using classical concentration arguments we prove in Appendix~\ref{app:proof_ce_and_lambda_finite} the following upper bounds.
\begin{lemma}[Bounded subsets of times]For $ 0 < \epsilon < \epsilon_\nu$ and $\gamma \in (0,1/2)$, 
 \[\Esp_\nu[\abs{\cE_{\epsilon,\gamma}^c}] \leq \dfrac{17}{\gamma\epsilon^4}  \abs{\cA}^2\abs{\cB}^2 \qquad		\Esp_\nu[\abs{\Lambda_{\epsilon_\nu}}] \leq 2\abs{\cA}^2\abs{\cB}(1 + E_\nu)^{\abs{\cB}} \,.
 \]
Refer to Appendix~\ref{notations} for the definitions of $\epsilon_\nu$ and $E_\nu$.

 \label{lem:ce_and_lambda_are_finite}
\end{lemma}
Thus combining them with \eqref{eq:decomp_T_epsilon} we obtain 
\[
\Esp_\nu[\abs{\cT_{\epsilon,\gamma}^c}] \leq \Esp_\nu[\abs{\cE_{\epsilon,\gamma}^c}] + \Esp_\nu[\abs{\Lambda_{\epsilon_\nu}}] \leq \dfrac{17}{\gamma\epsilon^4}  \abs{\cA}^2\abs{\cB}^2 +  2\abs{\cA}^2\abs{\cB}(1 + E_\nu)^{\abs{\cB}} \,.
\]
Hence, we just proved the following lemma.
\begin{lemma}[Reliable estimators] \label{imedstar reliability}For $ 0 < \epsilon < \epsilon_\nu$ and $\gamma \in (0,1/2)$, 
\[
\Esp_\nu[\abs{\cT_{\epsilon,\gamma}^{ c}}] \leq \dfrac{17}{\gamma\epsilon^4}  \abs{\cA}^2\abs{\cB}^2 +  2\abs{\cA}^2\abs{\cB}(1 + E_\nu)^{\abs{\cB}} \,.
\]
Refer to Appendix~\ref{notations} for the definitions of $\epsilon_\nu$ and $E_\nu$.
\end{lemma}
\subsection{Pareto-optimality and upper bounds on the numbers of pulls of sub-optimal arms}
In this section, we combine the different results of the previous sections to prove the following proposition.
\begin{proposition}[Upper bounds] \label{prop:upper bounds} Let $ \nu \in \cD_\omega$. Let $ 0 < \epsilon < \epsilon_\nu$ and $\gamma \in (0,1/2)$. Let us consider
\[
 \cT_{\epsilon, \gamma} \coloneqq \Set{t \geq 1:\ \begin{array}{l}
		  \Ohat^\star(t) = \cO^\star  \\  \forall (a,b) \textnormal{ s.t. } N_{a,b}(t) \geq \gamma\, N_{a_{t+1},b_{t+1}}(t) \textnormal{ or } (a,b) \in \cO^\star, \ \abs{\muhatab(t) - \muab } < \epsilon 
		\end{array}   }\,. 
\] 
Then under \IMEDSstar strategy,
\[
\Esp_\nu[\abs{\cT_{\epsilon,\gamma}^{ c}}] \leq \dfrac{17}{\gamma\epsilon^4}  \abs{\cA}^2\abs{\cB}^2 +  2\abs{\cA}^2\abs{\cB}(1 + E_\nu)^{\abs{\cB}} 
\]
and for all horizon time $ T \geq 1$,
\beqan 
\forall a \in \cA,\ \min\limits_{b:\, (a,b) \notin \cO^\star }\dfrac{1}{\log\!\left(N_b(T)\right)} \sum\limits_{b'\in \Bab}N_{a,b'}(T)\, \kl(\mu_{a,b'}|\mu_b^\star - \omegabb) 
\!\!\!&\leq& \!\!\!(1 + \alpha_\nu(\epsilon))\left[ 1 + \gamma \dfrac{M_\nu}{m_\nu} \right]  \\
&+&\!\!\! \dfrac{M_\nu\abs{\cT_{\epsilon,\gamma}}}{\min_{b \in \cB}\log\!\left(N_b(T)\right)}   
\eeqan
where $m_\nu$ and $M_\nu$ are defined  as follows:
\[
m_\nu = \min\limits_{\substack{(a,b) \notin \cO^\star \\ b'\in \Bab}} \klabb,\qquad M_\nu = \max\limits_{(a,b) \notin \cO^\star}\sum\limits_{b'\in \Bab} \klabb  \,.
\]
Furthermore, we have 
\[ \forall (a,b) \notin \cO^\star,\quad N_{a,b}(T) \leq \dfrac{1 + \alpha_\nu(\epsilon)}{m_\nu} \log\!\left(N_b(T)\right) + \abs{\cT_{\epsilon,\gamma}^c} \,.
\]
Refer to Appendix~\ref{notations} for the definitions of $\epsilon_\nu$, $\alpha_\nu(\cdot)$ and $E_\nu$.
\end{proposition}
\begin{proof}From Lemma~\ref{imedstar reliability}, we have: 
		\[
		\Esp_\nu[\abs{\cT_{\epsilon,\gamma}^{ c}}] \leq \dfrac{17}{\gamma\epsilon^4}  \abs{\cA}^2\abs{\cB}^2 +  2\abs{\cA}^2\abs{\cB}(1 + E_\nu)^{\abs{\cB}} \,.
		\]
Let $a \in \cA$. Let us consider $ 1 \leq t \leq T $ such that $a_{t+1} = a$, $(a_{t+1},b_{t+1}) \notin \cO^\star$ and $ t \in \cT_{\epsilon,\gamma}$. Then, according to \IMEDSstar strategy (see Algorithm~\ref{alg:IMEDSstar}), we have $$ (a,\bul_t) = (\aul_t,\bul_t) \notin \Ohat^\star(t) \,.$$
From Lemma~\ref{imedstar  empirical upper bounds} this implies
\beq \label{eq1_proof_pareto_opt}
\sum\limits_{b'\in \Bhat_{a,\bul_t}(t)} N_{a,b'}(t)\, \klof{\muhat_{a,b'}(t)}{\muhat_{\bul_t}^\star(t) - \omega_{\bul_t,b'}} \leq \log\!\left(N_b(T)\right)\,.
\eeq
Since $t \in \cT_{\epsilon,\gamma}$ and $\epsilon < \epsilon_\nu$, we  have 
\beq \label{incl1_proof_pareto_opt}
\Set{b ' \in \cB_{a,\bul_t}:\ N_{a,b'}(t) \geq \gamma N_{a_{t+1},b_{t+1}}(t)} \subset \Bhat_{a,\bul_t}(t) \,. 
\eeq
Combining inequality~(\ref{eq1_proof_pareto_opt}) with inclusion~(\ref{incl1_proof_pareto_opt}), it comes
\beq \label{eq2_proof_pareto_opt}
\sum\limits_{b'\in \cB_{a,\bul_t}:\ N_{a,b'}(t) \geq\gamma N_{a_{t+1},b_{t+1}}(t)} N_{a,b'}(t)\, \klof{\muhat_{a,b'}(t)}{\muhat_{\bul_t}^\star(t) - \omega_{\bul_t,b'}} \leq \log\!\left(N_b(T)\right) \,.
\eeq
Since $t \in \cT_{\epsilon,\gamma}$, we have 
\beq \label{eq3_proof_pareto_opt}
\abs{\muhat_{\bul_t}^\star(t) - \mu_{\bul_t}^\star} < \epsilon \ad \forall b'\in \cB_{a,\bul_t} \textnormal{ s.t. }N_{a,b'}(t) \geq\gamma N_{a_{t+1},b_{t+1}}(t),\ \abs{\muhat_{a,b'}(t) - \mu_{a,b'}} < \epsilon \,.
\eeq
By construction of $\alpha_\nu(\cdot)$ (see Section~\ref{notations}), since $ \epsilon < \epsilon_\nu$, inequalities~(\ref{eq2_proof_pareto_opt}) and~(\ref{eq3_proof_pareto_opt}) give us
\beq \label{eq4_proof_pareto_opt}
\sum\limits_{b'\in \cB_{a,\bul_t}:\ N_{a,b'}(t) \geq\gamma N_{a_{t+1},b_{t+1}}(t)} N_{a,b'}(t)\, \klof{\mu_{a,b'}}{\mu_{\bul_t}^\star - \omega_{\bul_t,b'}} \leq \left(1 + \alpha_\nu(\epsilon))\right)\log\!\left(N_b(T)\right) \,.
\eeq
This implies
\[
\sum\limits_{b'\in \cB_{a,\bul_t}} N_{a,b'}(t)\, \klof{\mu_{a,b'}}{\mu_{\bul_t}^\star - \omega_{\bul_t,b'}} \leq \left(1 + \alpha_\nu(\epsilon))\right)\log\!\left(N_b(T)\right) + \gamma M_\nu N_{a_{t+1},b_{t+1}}(t) \,.
\]
Furthermore, using inequality~(\ref{eq4_proof_pareto_opt}), we get
\beq \label{eq5_proof_pareto_opt}
N_{a_{t+1},b_{t+1}}(t) \leq \left\{ \begin{array}{ll}
    N_{a,\bul_t}(t) \leq \dfrac{\left(1+\alpha_\nu(\epsilon)\right)\log\!\left(N_b(T)\right)}{\klof{\mu_{a,\bul_t}}{\mu_{\bul_t}^\star}} & \textnormal{if } c_a = c_a^+ \,,\\
     \dfrac{\left(1+\alpha_\nu(\epsilon)\right)\log\!\left(N_b(T)\right)}{\klof{\mu_{a_{t+1},b_{t+1}}}{\mu_{\bul_t}^\star -\omega_{\bul_t,b_{t+1}}}} & \textnormal{if } c_a < c_a^+ \,. 
\end{array} \right.
\eeq

~\\ Thus, we have shown that for all arm $a \in \cA$, for all time step $1 \leq t \leq T$ such that $a_{t+1} = a$, $(a_{t+1},b_{t+1}) \notin \cO^\star$ and $ t \in \cT_{\epsilon,\gamma}$:
\[
\min\limits_{b:\, (a,b) \notin \cO^\star } \dfrac{1}{\log\!\left(N_b(T)\right)}\sum\limits_{b'\in \Bab}N_{a,b'}(t)\, \kl(\mu_{a,b'}|\mu_b^\star - \omegabb) \leq \left(1 + \alpha_\nu(\epsilon))\right)\left(1 + \gamma \dfrac{M_\nu}{m_\nu}\right)
\]
and 
\[
N_{a_{t+1},b_{t+1}}(t) \leq 
     \dfrac{\left(1+\alpha_\nu(\epsilon)\right)\log\!\left(N_b(T)\right)}{m_\nu} \,.
\]
This implies for all arm $a \in \cA$ and for all time step $ 1 \leq t \leq T$,
\beqan
\min\limits_{b:\, (a,b) \notin \cO^\star }\dfrac{1}{\log\!\left(N_b(T)\right)} \sum\limits_{b'\in \Bab}N_{a,b'}(T)\, \kl(\mu_{a,b'}|\mu_b^\star - \omegabb) 
\!\!\!&\leq& \!\!\!(1 + \alpha_\nu(\epsilon))\left[ 1 + \gamma \dfrac{M_\nu}{m_\nu} \right]  \\
&+&\!\!\! \dfrac{M_\nu\abs{\cT_{\epsilon,\gamma}}}{\min_{b \in \cB}\log\!\left(N_b(T)\right)}   
\eeqan
and 
\[
\forall b: (a,b) \notin \cO^\star,\quad N_{a,b}(T) \leq 
     \dfrac{\left(1+\alpha_\nu(\epsilon)\right)\log\!\left(N_b(T)\right)}{m_\nu} + \abs{\cT_{\epsilon,\gamma}^c}\,.
\]
\end{proof}
It can be easily proved that under \IMEDSstar $N_b(T) \to \infty$ for all $b \in \cB$ (see Lemma~\ref{nb pulls to infty}). From previous Proposition~\ref{prop:upper bounds}, we deduce the following corollary by doing $T \to \infty$, then $\epsilon,\,\gamma \to 0$.
\begin{corollary}[Pareto optimality]\label{cor: pareto optimality}Let $\nu \in \cD_\omega$. Let $ a\!\in\! \cA$ such that $\Set{b \in \cB: (a,b) \notin \cO^\star} \neq \emptyset$. Then, we have 
$$ \limsupT\min\limits_{b:\, (a,b) \notin \cO^\star }\dfrac{1}{\log\!\left(N_b(T)\right)} \sum\limits_{b'\in \Bab}N_{a,b'}(T) \kl(\mu_{a,b'}|\mu_b^\star - \omegabb)  \leq  1 \,.$$
\end{corollary}

\subsection{\label{appendix uniformly spread}\texorpdfstring{\IMEDSstar}{TEXT} is consistent and induces sequences of users with log-frequencies \texorpdfstring{$1_\cB$}{TEXT}}
In this section we show that \IMEDSstar is a consistent strategy that induces sequences of users with log-frequencies all equal to $1$, independently from the considered bandit configuration in $\cD$.
\begin{lemma}[Consistency, log-frequencies $1_\cB$]\IMEDSstar is a consistent strategy and induces sequences of users with log-frequencies all equal to $1$.
\end{lemma}
\begin{proof}
We first show that  \IMEDSstar induces sequences of users with log-frequencies all equal to $1$.

~\\Let $ \nu \in \cD_\omega$ and let us consider an horizon $T \geq 1$. Let $ 0 < \epsilon < \epsilon_\nu$ and $\gamma \in (0,1/2)$. Let us consider again the set of times
\[
 \cT_{\epsilon, \gamma} \!=\! \Set{\!T\!\geq\! t \!\geq\! 1\!: \begin{array}{l}
		  \Ohat^\star(t) = \cO^\star  \\  \forall (a,b) \textnormal{ s.t. } N_{a,b}(t) \geq \gamma\, N_{a_{t+1},b_{t+1}}(t) \textnormal{ or } (a,b) \in \cO^\star, \ \abs{\muhatab(t) - \muab } < \epsilon 
		\end{array}  }. 
\] 
Then, according to Proposition~\ref{prop:upper bounds}, under \IMEDSstar~\ strategy,
\beq \label{eq:freq_1}
\Esp_\nu[\abs{\cT_{\epsilon,\gamma}^{ c}}] \leq \dfrac{17}{\gamma\epsilon^4}  \abs{\cA}^2\abs{\cB}^2 +  2\abs{\cA}^2\abs{\cB}(1 + E_\nu)^{\abs{\cB}} < \infty \,.
\eeq
and for all horizon time $ T \geq 1$, for all $(a,b) \!\notin\!\cO^\star$,
\beq \label{eq:freq_2}
N_{a,b}(T)\leq \dfrac{1 + \alpha_\nu(\epsilon)}{m_\nu}\log\!\left(N_b(T)\right) + \abs{\cT_{\epsilon,\gamma}^c} \leq \dfrac{1 + \alpha_\nu(\epsilon)}{m_\nu}\log(T) + \abs{\cT_{\epsilon,\gamma}^c} \,,
\eeq
where $m_\nu = \min\limits_{\substack{(a,b) \notin \cO^\star \\ b'\in \Bab}} \klabb $, and $\epsilon_\nu$, $\alpha_\nu(\cdot)$, $E_\nu$ defined in Appendix~\ref{notations}.

~\\ Note that, under \IMEDSstar, for all $ t \geq 1$ such that $(a_{t+1},b_{t+1}) \in \Ohat^\star(t)$ we have 
~\\$ (a_{t+1},b_{t+1}) \in \argmin\limits_{(a,b) \in \Ohat^\star(t)}N_{a,b}(t) $.
This implies by definition of $\cT_{\epsilon,\gamma}$  that
\beq \label{eq:freq_3}
\forall (a,b),\, (a',b') \in \cO^\star,\quad \abs{N_{a,b}(T) - N_{a',b'}(T) } \leq \abs{\cT_{\epsilon,\gamma}^c} + 1 \,.
\eeq
Indeed the difference of pulls between two optimal couples is non-decreasing only at times $t\geq1$ such that the difference is greater than $1$ and $\Ohat^\star(t) \neq \cO^\star$. Combining Eq. \ref{eq:freq_2} and \ref{eq:freq_3} we get
\beqan
\minb N_b(T) 
&\geq& \min\limits_{(a,b) \in \cO^\star}N_{a,b}(T) -1\\
&\geq& \max\limits_{(a,b) \in \cO^\star}N_{a,b}(T) -  \abs{\cT_{\epsilon,\gamma}^c} - 1 \qquad \text{(Eq. \ref{eq:freq_2})} \\
&\geq& \dfrac{1}{\abs{\cB}}\sum\limits_{(a,b) \in \cO^\star}N_{a,b}(T) - \abs{\cT_{\epsilon,\gamma}^c} -1\\
&=& \dfrac{1}{\abs{\cB}} \left( T - \sum\limits_{(a,b) \notin \cO^\star}N_{a,b}(T)\right) - \abs{\cT_{\epsilon,\gamma}^c} - 1\\
&\geq& \dfrac{1}{\abs{\cB}} \left( T - (\abs{\cA}-1)\abs{\cB}\left[ \dfrac{1 + \alpha_\nu(\epsilon)}{m_\nu}\log(T) + \abs{\cT_{\epsilon,\gamma}^c} \right] \right) - \abs{\cT_{\epsilon,\gamma}^c} - 1 \qquad \text{(Eq. \ref{eq:freq_3})}\\
&\geq& \dfrac{T}{\abs{\cB}}  - \abs{\cA} \dfrac{1 + \alpha_\nu(\epsilon)}{m_\nu}\log(T) - \abs{\cA} \abs{\cT_{\epsilon,\gamma}^c} - \abs{\cA} \,. \\
\eeqan
Since $\Esp_\nu\!\left[\abs{\cT_{\epsilon,\gamma}^c} \right]< \infty$ (see Eq. \ref{eq:freq_1}), this implies that \IMEDSstar induces sequences of users with log-frequencies all equal to $1$.

~\\We show the consistency of \IMEDSstar in the following. Let $(a,b) \notin \cO^\star$ and $ \alpha \in (0,1)$.  According to Proposition~\ref{prop:upper bounds},
$$  N_{a,b}(T)\leq \dfrac{1 + \alpha_\nu(\epsilon)}{m_\nu}\log\!\left(N_b(T)\right) + \limsupT\abs{\cT_{\epsilon,\gamma}^c} \,,$$
and monotone convergence theorem ensures
\[\Esp_\nu[\limsupT \abs{\cT_{\epsilon,\gamma}^{c}}]
=\limsupT\Esp_\nu[\abs{\cT_{\epsilon,\gamma}^{ c}}] \leq \dfrac{17}{\gamma\epsilon^4}  \abs{\cA}^2\abs{\cB}^2 +  2\abs{\cA}^2\abs{\cB}(1 + E_\nu)^{\abs{\cB}} < \infty \,. 
\]
This implies 
\[
\dfrac{N_{a,b}(T)}{N_b(T)^\alpha} 
\leq \dfrac{1 + \alpha_\nu(\epsilon)}{m_\nu}\dfrac{\log\!\left(N_b(T)\right)}{N_b(T)^\alpha} + \dfrac{\limsupT\abs{\cT_{\epsilon,\gamma}^c}}{N_b(T)^\alpha} \,,
\] 
and, taking the expectation, dominated convergence theorem implies 
\[
\Esp_\nu\!\left[\dfrac{N_{a,b}(T)}{N_b(T)^\alpha}\right] \leq \Esp_\nu\!\left[\dfrac{1 + \alpha_\nu(\epsilon)}{m_\nu}\dfrac{\log\!\left(N_b(T)\right)}{N_b(T)^\alpha} + \dfrac{\limsupT\abs{\cT_{\epsilon,\gamma}^c}}{N_b(T)^\alpha}\right] \to 0 \,. \]
Indeed, it can be easily shown that under \IMEDSstar $N_b(T) \to \infty$ (see Lemma~\ref{nb pulls to infty}). This implies
$$ \limsupT \Esp_\nu\left[\dfrac{N_{a,b}(T)}{N_b(T)^\alpha}\right] = 0 \,. $$
\end{proof}
\subsection{The counters \texorpdfstring{$c_a$}{TEXT} and \texorpdfstring{$c_a^+$}{TEXT} coincide at most  \texorpdfstring{$O\left(\log(\log(T))\right)$}{TEXT} times}
\label{app: c c+}
Let us consider $0 < \epsilon < \epsilon_\nu $ and $\gamma \in (0,1/2)$. Let us introduce 
$$ \cT_c(T) \coloneqq \Set{ t \in \cT_{\epsilon,\gamma}:\ (a_{t+1},b_{t+1}) \notin \cO^\star \textnormal{ and } c_{a_{t+1}}(t) = c_{a_{t+1}}^+(t) }  \,, $$
where $\cT_{\epsilon,\gamma}$ is define as in Appendix~\ref{subsec: reliable estimators}.

~\\ In this section, we want to bound $\abs{\cT_c(T)}$.

\begin{lemma} \label{lem:Tc(T)} Let $ 0 < \epsilon < \epsilon_\nu $ and $\gamma \in (0,1/2)$. Let us consider an horizon $T \geq 1$. Then, it holds
 \beqan \abs{\cT_c(T)}  \leq 2\abs{\cA} + \abs{\cA} \log_2\left(\dfrac{(1 + \alpha_\nu(\epsilon))\abs{\cB}}{m_\nu}\log(T) + \abs{\cB}\abs{\cT_{\epsilon,\gamma}^c}\right)  \,.\eeqan
 where $m_\nu = \min\limits_{\substack{(a,b) \notin \cO^\star \\ b'\in \Bab}} \klabb $, and $\epsilon_\nu$, $\alpha_\nu(\cdot)$ defined in Appendix~\ref{notations}.
\end{lemma}
\begin{proof}From Lemma~\ref{lem: T_c and log(c_a)}, we get:
$$ \abs{\cT_c(T)} \leq 2\abs{\cA} +  \suma \log_2(c_a(T)) \,.$$
Then applying Lemma~\ref{counter - Na}, it comes:
$$ \abs{\cT_c(T)} \leq 2\abs{\cA} +  \suma \log_2\left(\sum\limits_{b:\ (a,b) \notin \cO^\star}N_{a,b}(T) + \abs{\cT_{\epsilon,\gamma}^c}\right) \,.$$
We end the proof by combining the previous inequality with Proposition~\ref{prop:upper bounds} that ensures
$$ \forall (a,b) \notin \cO^\star,\quad N_{a,b}(T) \leq \dfrac{\left(1+\alpha_\nu(\epsilon)\right)}{m_\nu}\log(T) + \abs{\cT_{\epsilon,\gamma}^c} \,. $$
\end{proof}

\begin{lemma} \label{lem: T_c and log(c_a)}Let  $ 0 < \epsilon < \epsilon_\nu $ and $\gamma \in (0,1/2)$. Let us consider an horizon $T \geq 1$. Then, it holds
$$ \abs{\cT_c(T)} \leq 2\abs{\cA} +  \suma \log_2(c_a(T)) \,.$$
Refer to Appendix~\ref{notations} for the definition of $\epsilon_\nu$.
\end{lemma} 
\begin{proof}  Let $a \in \cA$. By construction of $(c_a(t))_{1\leq t\leq T}$ and $ (c_a^+(t))_{1\leq t\leq T} $, we have
\beq \label{eq:c=c_1}
c_a^+(T) = 2^{\sum\limits_{t = 1}^T \ind_{\Set{ c_a(t) =c_a^+(t) }} -1} \,.
\eeq
Furthermore, the following inequalities are satisfied
\beq \label{eq:c=c_2}
\abs{\cT_c(T)} \leq \suma \sum\limits_{t = 1}^T \ind_{\Set{ c_a(t) =c_a^+(t) }} \ad \forall a\in \cA,\  c_a^+(T) \leq 2c_a(T) .
\eeq
Then Eq. \ref{eq:c=c_1} and \ref{eq:c=c_2} imply 
$$ 2^{\abs{\cT_c(T)}-\abs{\cA}} \leq 2^{\abs{\cA}}\prod\limits_{a\in\cA}c_a(T) \,. $$
\end{proof}

\begin{lemma}\label{counter - Na}Let  $ 0 < \epsilon < \epsilon_\nu $ and $\gamma \in (0,1/2)$. Let us consider an horizon $T \geq 1$. Then, it holds
$$ \forall a \in \cA,\quad \abs{c_a(T) - \sum\limits_{b:\ (a,b) \notin \cO^\star}N_{a,b}(T)} \leq  \abs{\cT_{\epsilon,\gamma}^c}  \,.$$
Refer to Appendix~\ref{notations} for the definition of $\epsilon_\nu$.
\end{lemma}
\begin{proof} Let $ a \in \cA$. At each time step $ t \geq 1$ we increment $c_a(t)$ only if $(a_{t+1},b_{t+1}) \notin \Ohat^\star(t) $ and $a_{t+1} = a$. Then, if $t \in \cT_{\epsilon,\gamma} $ , we have $\Ohat^\star(t) = \cO^\star$ and we increment $c_a(t) $ only if we increment one of the $ N_{a,b}(t)$ for $b \in \cB$ such that $(a,b) \notin \cO^\star$.
\end{proof}

\subsection{All couples \texorpdfstring{$(a,b)\!\in\!\cA\!\times\!\cB$}{TEXT} are asymptotically pulled an infinite number of times}
Let $ 0 < \epsilon < \epsilon_\nu$ (defined in Appendix~\ref{notations}) and $\gamma \in (0,1/2)$. Let us consider
\[
 \cT_{\epsilon, \gamma} = \Set{t \geq 1:\ \begin{array}{l}
		  \Ohat^\star(t) = \cO^\star  \\  \forall (a,b) \textnormal{ s.t. } N_{a,b}(t) \geq \gamma\, N_{a_{t+1},b_{t+1}}(t) \textnormal{ or } (a,b) \in \cO^\star, \ \abs{\muhatab(t) - \muab } < \epsilon 
		\end{array}   }\,. 
\] 
Then, according to Proposition~\ref{prop:upper bounds}, under \IMEDSstar strategy, $\Esp_\nu[\abs{\cT_{\epsilon,\gamma}^{ c}}]  < \infty$. In particular, almost surely $\abs{\cT_{\epsilon,\gamma}^c} < \infty$.

\begin{lemma}[The indexes tend to infinity]\label{The indexes tend to infinity} For all strategy we have $ \limt  N_{a_{t+1},b_{t+1}}(t) = \infty $
and, under \IMEDSstar, $$ \forall (a,b) \in \cA\times\cB,\ \limt I_{a,b}(t) = \infty \,.$$
\end{lemma}
\begin{proof}
~\\For all couple $(a,b) \in \cA\times\cB$ such that $\Nab(\infty) < \infty$, we have $ \ind_{\Set{(a_{t+1},b_{t+1}) = (a,b) }} \to 0 $. Then  $$ \sum\limits_{(a,b)\in\cA\times\cB:\ \Nab(\infty) = \infty} \ind_{\Set{(a_{t+1},b_{t+1}) = (a,b) }} \to 1  \,. $$
This implies
\beqan 
&&\sum\limits_{(a,b)\in\cA\times\cB:\ \Nab(\infty) = \infty} \ind_{\Set{(a_{t+1},b_{t+1}) = (a,b) }} N_{a,b}(t) \\ &\geq& \min\limits_{(a,b)\in\cA\times\cB:\ \Nab(\infty) = \infty}N_{a,b}(t) \sum\limits_{(a,b)\in\cA\times\cB:\ \Nab(\infty) = \infty} \ind_{\Set{(a_{t+1},b_{t+1}) = (a,b) }}  \longrightarrow \infty \,.
\eeqan 
Thus, since 
\beqan N_{a_{t+1},b_{t+1}}(t) 
&=&  \sum\limits_{(a,b)\in\cA\times\cB:\ \Nab(\infty) < \infty} \ind_{\Set{(a_{t+1},b_{t+1}) = (a,b) }} N_{a,b}(t) \\
&+& \sum\limits_{(a,b)\in\cA\times\cB:\ \Nab(\infty) = \infty} \ind_{\Set{(a_{t+1},b_{t+1}) = (a,b) }} N_{a,b}(t)  \,,\eeqan
we have 
$$ N_{a_{t+1},b_{t+1}}(t)  \longrightarrow \infty \,.$$
~\\ Furthermore, under \IMEDSstar strategy we have
$$  \forall (a,b) \in \cA\times\cB,\quad I_{a,b}(t) \geq \log(N_{a_{t+1},b_{t+1}}(t)) \,,$$
which ends the proof.

\end{proof}

\begin{lemma}[The numbers of pulls tend to infinity] \label{nb pulls to infty} Under \IMEDSstar the numbers of pulls almost surely satisfy 
$$ \forall (a,b) \in \cA\times\cB, \Nab(T) \to \infty \,.$$
In particular, almost surely for all $(a,b) \in \cA\times\cB$, $\limT \muhatab(T) = \muab $.
\end{lemma}
\begin{proof}Lemma~\ref{The indexes tend to infinity} ensures $\limT I_{a,b}(T) = \infty$, for all $(a,b) \in \cA\times\cB$.

~\\Let $(a,b) \in \cO^\star$. Since $ \abs{ \cT_{\epsilon,\gamma}^c} < \infty $ and $ \forall T \in \cT_{\epsilon,\gamma},\, \Ohat^\star(T) = \cO^\star $, we have for all $(a,b) \in \cO^\star$
$$ \limT\log(\Nab(T)) = \limT I_{a,b}(T) = \infty.$$
Then, for all $(a,b) \in \cO^\star$, $\limT \Nab(T) = \infty$ and $\limT \muhatab(T) = \muab$.
 
~\\ Let $(a,b) \notin \cO^\star$ and let $T \in \cT_{\epsilon,\gamma}$. Then, the following inequalities occur
\beq \label{eq:Ninfty_1} 
I_{a,b}(T) \leq \sum\limits_{b' :\, (a,b') \notin \cO^\star }N_{a,b'}(t) \kl(\muhat_{a,b'}(t)|\muhat_b^\star(t) - \omegabb) + \log(\max(1,N_{a,b'}(t))) 
\eeq
and 
\beq \label{eq:Ninfty_2}
\forall b \in \cB,\quad \muhat_b^\star(t) < 1 - \epsilon_\nu \,.
\eeq
Since $\abs{\cT_{\epsilon,\gamma}^c} < \infty$, Eq. \ref{eq:Ninfty_1} and \ref{eq:Ninfty_2} imply
$$ \limT \sum\limits_{b':\, (a,b') \notin \cO^\star} N_{a,b'}(T) \to \infty   \,.$$
Then, since $\abs{\cT_{\epsilon,\gamma}^c} < \infty$, from Lemma~\ref{counter - Na} we get $\limT c_a(T) = \infty  $.
This implies 
$$ \argmax\Set{t\in\llbracket1,T\rrbracket:\ c_a(t) =c_a^+(t)} \to \infty \ad \limT \minb N_{a,b}(T)  = \infty \,.$$
\end{proof}

\subsection{\label{app: concentration lemma statements} Concentration lemmas}
	We state two concentration lemmas that do not depend on the followed strategy. Lemma~\ref{imedstar concentration} comes from Lemma B.1 in \citet{combes2014unimodal} and Lemma~\ref{imedstar large deviation} comes from Lemma 14 in \citet{honda2015imed}. Proofs are provided in Appendix~\ref{app: concentration_lemmas}.
	\label{subsec : imed_concentration}
	\begin{lemma}[Concentration inequalities]\label{imedstar concentration} Let  $\nu \!\in\! \cD_\omega $. For all $  0 \!<\! \epsilon, \gamma \!\leq\! 1/2 $ and  for all couples $(a,b),\, (a',b') \in \cA\!\times\!\cB$, 
		\[
		\Esp_\nu\left[\sum\limits_{t\geq 1}\ind_{\Set{(a_{t+1},b_{t+1})=(a,b),\ N_{a',b'}(t) \geq \gamma  N_{a,b}(t),\  \abs{\muhat_{a',b'}(t) - \mu_{a',b'}} \geq \epsilon }} \right] \leq \dfrac{17}{\gamma\epsilon^4} \,.
		\]
	\end{lemma}

	\begin{lemma}[Large deviation probabilities]\label{imedstar large deviation} Let  $\nu \!\in\! \cD_\omega $. For all couple $(a,b) \!\in\! \cA\!\times\!\cB$, for all $0 \!<\! \mu \!<\! \muab$\,,
		\[ 
		  \Esp_\nu\!\left[\!\sum\limits_{n \geq 1}\ind_{\Set{ \muhat_{a,b}^{n}  < \mu}}n\exp(n \kl(\muhat_{a,b}^{n}|\mu))\!\right]\! \leq\! 6\e\!\left(\!1\!-\! \frac{\log(1-\mu)}{\log(1-\muab)}\!\right)^{\text{-}1}\!\left(\!1\!-\!\e^{-\left(\!1- \frac{\log(1-\mu)}{\log(1-\muab)}\!\right)\!\kl\!(\muab|\mu)}\!\right)^{\text{-}3},
		\]
		where $\muhatab^n$ estimates $\muab$ after $n$ pulls of couple $(a,b)$ (see Appendix~\ref{notations}).
	\end{lemma}
	
\subsection{Proof of Lemma~\ref{lem:ce_and_lambda_are_finite}}
\label{app:proof_ce_and_lambda_finite}


~\\Using Lemma~\ref{imedstar empirical lower bounds}, for all time step $ t\geq 1$, we have
		\[
		\forall  (a',b') \in \Ohat^\star(t), \quad N_{a',b'}(t) \geq N_{a_{t+1},b_{t+1}}(t) \geq \gamma\,  N_{a_{t+1},b_{t+1}}(t) \,.
		\]
		Then, based on the concentration inequalities from Lemma~\ref{imedstar concentration},  we obtain
		\beqan 
		\Esp_\nu[\abs{\cE_{\epsilon,\gamma}^c}] 
		&\leq& \sum\limits_{(a,b), (a',b') \in\cA\times\cB} \Esp_\nu\left[\sum\limits_{t\geq1}\ind_{\Set{(a_{t+1},b_{t+1})=(a,b),\ N_{a',b'}(t) \geq \gamma  N_{a,b}(t),\  \abs{\muhat_{a',b'}(t) - \mu_{a',b'}} \geq \epsilon }}\right] \\
		&\leq & \sum\limits_{(a,b), (a',b') \in\cA\times\cB} \dfrac{17}{\gamma\epsilon^4} \\
		&\leq& \dfrac{17}{\gamma\epsilon^4}  \abs{\cA}^2\abs{\cB}^2 \,.
		\eeqan
Furthermore, for $t \geq 1$ , $a  \in \cA$ and $\cB'\subset\cB$, we have
		{
			\beqan
			&&  \log\left(N_{a_{t+1},b_{t+1}}(t)\right) \leq 
				\sum\limits_{b \in \cB'}\Nab(t) \, \klof{\muhatab(t)}{\lambda_{a,b} } + \log\left(\Nab(t)\right)  \\
			&\Leftrightarrow& N_{a_{t+1},b_{t+1}}(t) \leq 
				\prod\limits_{b \in \cB'}\Nab(t)\e^{\Nab(t) \, \klof{\muhatab(t)}{\lambda_{a,b} }}    \,,
			\eeqan
		}
		~\\where $\lambda_{a,b}  = \muab - \epsilon_\nu $ for all couple $(a,b) \in \cA\times\cB$. Thus, considering estimators of means based on the numbers of pulls $(\muhatab^n)_{(a,b)\in\cA\times\cB,n\geq1}$ (see Appendix~\ref{notations}),  we have
		{\footnotesize
			\beqan  \abs{\Lambda_{\epsilon_\nu}}
			&\leq&\sum\limits_{t \geq 1}\sum\limits_{ \substack{a \in \cA\\ \cB'\subset\cB}} \ind_{\Set{ \forall b \in\cB',\, \muhatab(t) < \lambdaab \textnormal{ and } N_{a_{t+1},b_{t+1}}(t) \leq 
				\prod\limits_{b \in \cB'}\Nab(t)\e^{\Nab(t) \, \kl(\muhatab(t)|\lambda_{a,b} )}  }} \\
			&=& \sum\limits_{\substack{ t \geq 1\\ (a',b')\in \cA\times\cB}}\sum\limits_{\substack{a \in \cA\\ \cB'\subset\cB}}\sum\limits_{\substack{ n_{b} \geq 0\\ b \in  \cB'}} \ind_{\Set{(a_{t+1},b_{t+1}) = (a',b'), \Nab(t) = n_b }} \ind_{\Set{ \forall b \in \cB',\, \muhatab^{n_b} < \lambdaab,\ N_{a',b'}(t) \leq \prod\limits_{b \in \cB'}n_b\e^{n_b \, \kl(\muhatab^{n_b}|\lambdaab )}   }} \\
			&\leq& \sum\limits_{\substack{ t \geq 1\\ (a',b')\in \cA\times\cB}}\sum\limits_{\substack{a \in \cA\\ \cB'\subset\cB}}\sum\limits_{\substack{ n_{b} \geq 1 \\ b \in  \cB'}} \ind_{\Set{(a_{t+1},b_{t+1}) = (a',b')}} \ind_{\Set{ \forall b \in \cB',\, \muhatab^{n_b} < \lambdaab}} \ind_{\Set{1 \leq N_{a',b'}(t) \leq \prod\limits_{b \in \cB'}n_b\e^{n_b \, \klof{\muhatab^{n_b}}{\lambdaab }}   }} \\
			&+& \sum\limits_{\substack{ t \geq 1\\ (a',b')\in \cA\times\cB}} \ind_{\Set{(a_{t+1},b_{t+1}) = (a',b')}} \ind_{\Set{N_{a',b'}(t) = 0  }} \\
			&\leq& \sum\limits_{\substack{ (a',b')\in \cA\times\cB}}\sum\limits_{\substack{a \in \cA\\ \cB'\subset\cB}}\sum\limits_{\substack{ n_{b} \geq 1 \\ b \in  \cB'}}  \ind_{\Set{ \forall b \in \cB',\, \muhatab^{n_b} < \lambdaab}} \sum\limits_{t \geq 1}\ind_{\Set{(a_{t+1},b_{t+1}) = (a',b')}}\ind_{\Set{1 \leq N_{a',b'}(t) \leq \prod\limits_{b \in \cB'}n_b\e^{n_b \, \kl(\muhatab^{n_b}|\lambdaab )}   }} \\
			&+& \abs{\cA}\abs{\cB} \\
			&\leq& \sum\limits_{\substack{ (a',b')\in \cA\times\cB}}\sum\limits_{\substack{a \in \cA\\ \cB'\subset\cB}}\sum\limits_{\substack{ n_{b} \geq 1 \\ b \in  \cB'}}  \ind_{\Set{ \forall b \in \cB',\, \muhatab^{n_b} < \lambdaab}} \prod\limits_{b \in \cB'}n_b\e^{n_b \, \kl(\muhatab^{n_b},\lambdaab )} +  \abs{\cA}\abs{\cB} \\
			&=& \abs{\cA}\abs{\cB}\sum\limits_{\substack{a \in \cA\\ \cB'\subset\cB}}\sum\limits_{\substack{ n_{b} \geq 1 \\ b \in  \cB'}} \prod\limits_{b \in \cB'} \ind_{\Set{\muhatab^{n_b} < \lambdaab}} n_b\e^{n_b \, \klof{\muhatab^{n_b}}{\lambdaab }} +  \abs{\cA}\abs{\cB} \\
			&=& \abs{\cA}\abs{\cB}\left[1 + \sum\limits_{\substack{a \in \cA\\ \cB'\subset\cB}}\prod\limits_{b \in \cB'}\sum\limits_{\substack{ n \geq 1 }}  \ind_{\Set{\muhatab^{n} < \lambdaab}} n\e^{n \, \kl(\muhatab^{n}|\lambdaab)}  \right]  \\
			\eeqan
		}
		and
		\beq \label{eq:concentration_1}
		\Esp_\nu[\abs{\Lambda_{\epsilon_\nu}}] \leq  \abs{\cA}\abs{\cB}\left(1 + \sum\limits_{\substack{a \in \cA\\ \cB'\subset\cB}}\prod\limits_{b \in \cB'}\Esp_\nu\left[\sum\limits_{\substack{ n \geq 1 }}  \ind_{\Set{\muhatab^{n} < \lambdaab}} n\e^{n \, \kl(\muhatab^{n},\lambdaab)}\right]\right) \,.
		\eeq
		Then, by applying Lemma~\ref{imedstar large deviation} based on large deviation inequalities, we have
		\beq \label{eq:concentration_2}
		\forall (a,b) \in \cA\times\cB,\quad \Esp_\nu\left[\sum\limits_{\substack{ n \geq 1 }}  \ind_{\Set{\muhatab^{n} < \lambdaab}} n\e^{n \, \kl(\muhatab^{n},\lambdaab)}\right] \leq E_\nu \,,
		\eeq
		where $E_\nu = 6\e\,\maxab \left(1- \frac{\log(1-\lambdaab)}{\log(1-\muab)}\right)^{-1}\left(1-\e^{-(1- \frac{\log(1-\lambdaab)}{\log(1-\muab)})\kl(\muab|\lambdaab)}\right)^{-3}$.
		~\\By combining Eq. \ref{eq:concentration_1} and \ref{eq:concentration_2}, we conclude  that
		\[
		\Esp_\nu[\abs{\Lambda_{\epsilon_\nu}}] \leq \abs{\cA}\abs{\cB} \left(1 + \abs{\cA} (1 + E_\nu)^{\abs{\cB}}\right) \leq 2\abs{\cA}^2\abs{\cB}(1 + E_\nu)^{\abs{\cB}} \,.
		\]

\section{\texorpdfstring{\IMEDSstar}{TEXT}: Proof of Theorem~\ref{th:asymptotic_optimality_IMEDSstar} (main result)}
\label{app:proof_main_result}
In this section we prove the asymptotic optimality of \IMEDSstar strategy. The proof is based on the finite time analysis detailed in Appendix~\ref{app: IMEDSstar finite time analysis}.

\subsection{ Almost surely \texorpdfstring{$\nopt(T)$}{TEXT} tends to \texorpdfstring{$n^\nu$}{TEXT}  }
For $a \in \cA$ such that $\cB_a = \Set{b \in \cB:\ (a,b) \notin \cO^\star} \neq \emptyset$, let us define the linear programming
\beqan
C_{\omega,a}^\star(\nu):=& \min\limits_{n \in \Real_+^{\cB_a} }&  \sum\limits_{b \in \cB_a} n_{b} \deltaab   
\\
& s.t. & \forall b \in \cB_a : \quad \sum\limits_{b' \in \cB_a}\kl^+(\mu_{a,b'}|\mu_b^\star - \omegabb) n_{a,b'} \geq 1 \,.
\eeqan
Then $(n^\nu_{a,b})_{b \in \cB_a}$ is the unique optimal solution of the previous minimization problem. Furthermore, we can state the following lemma.
\begin{lemma} \label{nopt nnu}$ \limT (\nopt_{a,b}(T))_{b \in \cB_a} = (n^\nu_{a,b})_{b \in \cB_a}$.
\end{lemma}
\begin{proof} This a direct application of Lemma~\ref{lemma: hyp continuity lp} and Lemma~\ref{lemma : continuity lp} stated below.

\end{proof}

\begin{lemma} \label{lemma: hyp continuity lp}Let $0 < \epsilon < \epsilon_\nu$ (see Appendix~\ref{notations}) and $\gamma \in (0,1)$. Let $ a \in \cA$ such that $\cB_a = \Set{b \in \cB:\ (a,b) \notin \cO^\star}$. Let us consider for $ T \geq 1$, $ \Khat_a(T) = ( \kl^+(\muhat_{a,b'}(T)|\muhat_b^\star(T) -\omegabb))_{b,b' \in \cB_a}$, the vector $ \Deltahat_a(T) = (\muhat_b^\star(T) - \muhat_{a,b}(T) )_{b \in \cB_a} $ and the parameter  $ \hhat_a(T) = (\Khat_a(T),\Deltahat_a(T) ) .$
We also consider 
$$ \cH_a \coloneqq \Set{ \hhat_a(T) \quad, T \in \cT_{\epsilon,\gamma}} \,,$$
where $\cT_{\epsilon,\gamma}$, defined in Appendix~\ref{subsec: reliable estimators}, satisfies $\abs{\cT_{\epsilon,\gamma}^c} < \infty$.
Then, we have 
$$ \forall h = (K,\Delta) \in \cH_a, \ K \neq 0 \ad \min\limits_{h = (K,\Delta) \in \cH_a}\min\limits_{b \in \cB_a}\Delta_b > 0 .$$
\end{lemma}

\begin{proof} Let $h = (K,\Delta) \in \cH_a$. There exists $ T \in \cT_{\epsilon,\gamma}$ such that $ h = (K,\Delta) = \hhat_a(T) = (\Khat_a(T),\Deltahat_a(T) ) $. Since $ T \in \cT_{\epsilon,\gamma} $, we have $ \Ohat^\star(T) = \cO^\star  $. In particular for all $ b \in \cB_a$,\  $ K_{b,b} = \Khat_{a,b,b}(T) = \kl^+(\muhatab(T)|\muhat_b^\star(T)) > 0  $. 
Furthermore, we have $$ \min\limits_{b \in \cB_a}\Delta_b = \min\limits_{b \in \cB_a}\Deltahat_{a,b}(T) = \min\limits_{b \in \Bhat_a(T)} \muhat_b^\star(T) - \muhat_{a,b}(T) > 0 \,.$$ 
Lastly since  $\forall b \in \cB_a, \muhat_b^\star(\infty)  = \mu_b^\star,\, \muhatab(\infty) = \muab $, we have 
$$\min\limits_{b \in \cB_a}\Deltahat_{a,b}(T) \to \min\limits_{b \in \cB_a}\mu_b^\star - \muab > 0 $$
and
$$ \min\limits_{h = (K,\Delta) \in \cH_a}\min\limits_{b \in \cB_a}\Delta_b = \min\limits_{T \notin \cE\cup\cT_0}\min\limits_{b \in \cB_a}\Deltahat_{a,b}(T) > 0 \,.$$
\end{proof}

\subsection{ Almost surely and on expectation, for all sub-optimal couple \texorpdfstring{$ \dfrac{\Nab(T)}{\log(T)} $}{TEXT} tends to \texorpdfstring{$n_{a,b}^\nu $}{TEXT} }
Combining the upper bounds from the finite analysis and the asymptotic behaviour of $\nopt(\cdot)$, we prove the asymptotic optimality of \IMEDSstar.  
\begin{lemma}[Asymptotic upper bounds]\label{asymptotic upper bound}For all sub-optimal couple $(a,b) \notin \cO^\star$,
$$   \limsupT \dfrac{N_{a,b}(T) }{\log(T)} \leq n^\nu_{a,b}  \,.$$
\end{lemma}
\begin{proof}Let $0 < \epsilon < \epsilon_\nu $ (see Appendix~\ref{notations}) and $\gamma \in (0,1/2)$. Let $(a,b) \notin \cO^\star$ and let us consider an horizon $T \geq 1$. Let us introduce the random variable
$$ \tau \coloneqq \min\Set{t \in \llbracket 1, T \rrbracket \st t \in \cT_{\epsilon,\gamma}\setminus\cT_c(T) \ad (a_{t+1},b_{t+1}) =(a,b) } \,,$$
where $\cT_{\epsilon,\gamma}$ and $\cT(T)$ are respectively introduced in Appendix~\ref{subsec: reliable estimators} and~\ref{app: c c+} . Then, by definition of $\tau$ and since $\abs{\cT_{\epsilon,\gamma}^c} < \infty$, from Lemma~\ref{lem:Tc(T)} we have 
\beq \label{eq:UB_1}
N_{a,b}(T) \leq N_{a,b}(\tau) + \abs{\cT_{\epsilon,\gamma}^c} + \abs{\cT_c(T)} = N_{a,b}(\tau) + O\left(\log(\log(T))\right) .
\eeq
Furthermore, since $ \tau \notin \cT_c(T) $ we have $c_a(\tau) \neq c_a^+(\tau)$. In addition, since $ \tau \in \cT_{\epsilon,\gamma}$ and $(a,b) = (a_{\tau + 1},b_{\tau + 1}) \notin \cO^\star$,  Lemma~\ref{imedstar empirical upper bounds} implies the following empirical upper bound
\beq \label{eq:UB_2}
N_{a,b}(\tau) \leq \log(\tau)\nopt_{a,b}(\tau) \,.
\eeq
In particular, since $\log(\tau)\leq \log(T)$, Eq. \ref{eq:UB_1} and \ref{eq:UB_2} imply
$$ a.s. \quad  \dfrac{N_{a,b}(T)}{\log(T) } \leq \dfrac{N_{a,b}(\tau)}{\log(\tau)} + \dfrac{O\left(\log(\log(T))\right)}{\log(T)}\leq \nopt(\tau) + \dfrac{O\left(\log(\log(T))\right)}{\log(T)} $$
and, since $a.s.$ $\limT\tau =\infty$, from Lemma~\ref{nopt nnu} we get
$$ a.s \quad \limsupT \dfrac{N_{a,b}(T)}{\log(T) } \leq \limsupT  \nopt_{a,b}(\tau)  + \limsupT \dfrac{O\left(\log(\log(T))\right)}{\log(T)} = n_{a,b}^\nu \,.$$

\end{proof}

\begin{lemma}[Asymptotic optimality]For all sub-optimal couple $(a,b) \notin \cO^\star$, we have $$  a.s. \  \limT \dfrac{N_{a,b}(T) }{\log(T)} = n^\nu_{a,b} \ad  \limT \Esp_\nu\!\left[\dfrac{N_{a,b}(T) }{\log(T)}\right] = n^\nu_{a,b}  \,.$$
\end{lemma}
\begin{proof}Since \IMEDSstar  is a consistent strategy that induces sequences of users with log-frequencies equal to $1$, we have
$$ \forall (a,b)\neq \cO^\star, \quad  \liminfT \dfrac{1}{\log(T)}  \sum\limits_{b' \in\cB_a} N_{a,b'}(T) \kl^+(\mu_{a,b'}|\mu_b^\star - \omegabb) \geq 1 \,.$$
Then, Pareto-optimality of \IMEDSstar combined with asymptotic upper bounds given in  Lemma~\ref{asymptotic upper bound} ensures that for all $(a,b) \notin \cO^\star$, $N_{a,b}(T)/\log(T) \to n^\nu_{a,b} $. Since, the $ N_{a,b}(T)/\log(T) $ are dominated by an integrable variable (see Proposition~\ref{prop:upper bounds} in Appendix~\ref{app: IMEDSstar finite time analysis}), we also have these convergences on average.
\end{proof}

\section{Concentration lemmas: Proofs}
	\label{app: concentration_lemmas}
	{\bf Lemma} Let  $\nu \!\in\! \cD_\omega $. For all $  0 \!<\! \epsilon, \gamma \!\leq\! 1/2 $ and  for all couples $(a,b),\, (a',b') \in \cA\!\times\!\cB$, 
		\[
		\Esp_\nu\left[\sum\limits_{t\geq 1}\ind_{\Set{(a_{t+1 },b_{t+1})=(a,b),\ N_{a',b'}(t) \geq \gamma  N_{a,b}(t),\  \abs{\muhat_{a',b'}(t) - \mu_{a',b'}} \geq \epsilon }} \right] \leq \dfrac{17}{\gamma\epsilon^4} \,.
		\]
	\begin{proof} Considering the stopping times $ \tau_{a,b}^n = \inf{ \Set{t \geq 1, N_{a,b}(t) = n} }$  we will rewrite the sum  
		~\\$ \sum\limits_{t\geq 1}\ind_{\Set{(a_{t+1},b_{t+1})=(a,b),\ N_{a',b'}(t) \geq \gamma  N_{a,b}(t),\  \abs{\muhat_{a',b'}(t) - \mu_{a',b'}} \geq \epsilon }} $ and use an Hoeffding's type argument.
		\beqan
		&& \sum\limits_{t \geq 1}\ind_{\Set{(a_{t+1},b_{t+1})=(a,b),\  N_{a',b'}(t) \geq \gamma N_{a,b}(t),\  \abs{\muhat_{a',b'}(t) - \mu_{a',b'}} \geq \epsilon }} \\
		&\leq& \sum\limits_{t \geq 1}\sum\limits_{n \geq 1,\, m \geq 0}\ind_{\Set{ \tau_{a,b}^n = t +1, N_{a',b'}(t) = m }} \ind_{\Set{ m \geq  \gamma(n-1) ,\  \abs{\muhat_{a',b'}^m - \mu_{a',b'}} \geq \epsilon }} \\
		&=& \sum\limits_{ m \geq 0} \sum\limits_{n \geq 1}\ind_{\Set{ m \geq \gamma(n-1),\  \abs{\muhat_{a',b'}^m - \mu_{a',b'}} \geq \epsilon }}\sum\limits_{t \geq 1}\ind_{\Set{ \tau_{a,b}^n = t +1, N_{a',b'}(t) = m }}  \\
		&\leq& \sum\limits_{ m \geq 0} \sum\limits_{n \geq 1}\ind_{\Set{ m \geq \gamma(n-1),\  \abs{\muhat_{a',b'}^m - \mu_{a',b'}} \geq \epsilon }}\sum\limits_{t \geq 1}\ind_{\Set{ \tau_{a,b}^n = t +1}}  \\
		&\leq& \sum\limits_{ m \geq 0} \sum\limits_{n \geq 1}\ind_{\Set{ m \geq  \gamma(n-1),\  \abs{\muhat_{a',b'}^m - \mu_{a',b'}} \geq \epsilon }} \\
		\eeqan 
		Taking the expectation , it comes:
		\beqan
		&&\Esp_\nu\left[\sum\limits_{t\geq 1}\ind_{\Set{(a_{t+1},b_{t+1})=(a,b),\ N_{a',b'}(t) \geq \gamma  N_{a,b}(t),\  \abs{\muhat_{a',b'}(t) - \mu_{a',b'}} \geq \epsilon }} \right]  \\
		&\leq& \sum\limits_{ m \geq 0} \sum\limits_{n \geq 1}\ind_{\Set{ m \geq \gamma(n-1) }} \Pr_\nu\left(  \abs{\muhat_{a'}^m - \mu_{a'}} \geq \epsilon  \right) \\
		&\leq& \sum\limits_{ m \geq 0} \sum\limits_{n \geq 1}\ind_{\Set{ m \geq \gamma(n-1) }} 2\e^{-2m\epsilon^2} \quad \textnormal{(Hoeffding inequality)} \\
		&=& 2\sum\limits_{ m \geq 0} \left(\dfrac{m}{\gamma}+1\right)\e^{-2m\epsilon^2} \\
		&=& 2\sum\limits_{ m \geq 1} \dfrac{m}{\gamma}\e^{-2m\epsilon^2} +  2\sum\limits_{ m \geq 0} \e^{-2m\epsilon^2} \\ \\
		&=& \dfrac{1}{\gamma} \dfrac{2\e^{-2\epsilon^2}}{(1 - \e^{-2\epsilon^2})^2} + \dfrac{1}{1 - \e^{-2\epsilon^2}} \\
		&=&  \dfrac{1}{\gamma} \dfrac{2\e^{2\epsilon^2}}{(\e^{2\epsilon^2} - 1)^2} + \dfrac{\e^{2\epsilon^2}}{ \e^{2\epsilon^2} - 1} \\ 
		&\leq& \dfrac{1}{\gamma}\dfrac{8\e^{1/2}}{ \epsilon^4} + \dfrac{\e^{2\epsilon^2}}{ 2\epsilon^2} \leq \dfrac{17}{\gamma\epsilon^4}\,.
		\eeqan
		
	\end{proof}
	
	~\\ {\bf Lemma} Let  $\nu \!\in\! \cD_\omega $. For all couple $(a,b) \in \cA\!\times\!\cB$, for all $0 < \mu < \muab$, 
		\[ 
		  \Esp_\nu\left[\sum\limits_{n \geq 1}\ind_{\Set{ \muhat_{a,b}^{n}  < \mu}}n\exp(n \kl(\muhat_{a,b}^{n}|\mu))\right] \leq \dfrac{6\e}{(1- \frac{\log(1-\mu)}{\log(1-\muab)})\left(1-\e^{-(1- \frac{\log(1-\mu)}{\log(1-\muab)})\kl(\muab|\mu)}\right)^3} \,,
		\]
		where $\muhatab^n$ estimates $\muab$ after $n$ pulls of couple $(a,b)$ (see Appendix~\ref{notations}).
	\begin{proof} The proof is based on a Chernoff type inequality and a calculation by measurement change.
	The proof comes from \citet{honda2015imed}. We explicit here the particular case of Bernoulli distributions for completeness.
	
	~\\Let us rephrase Proposition 11 from \cite{honda2015imed}. Since we consider Bernoulli distributions, we get a  more explicit formulation.  
\begin{proposition}\label{prop : large devdev}
Let  $\nu \in \cD_\omega$. Let $ (a,b) \in \cA\times\cB$ and $ 0 < \mu < \muab$. Then, for all $n \geq 0 $ and $ u \in \Real$, we have 
$$ \Pr_\nu( \kl(\muhatab^n|\mu) \geq u,\ \muhatab^n \leq \mu) \leq \left\{ \begin{array}{l}
    \e^{-n\kl(\muab|\mu)} \qquad \textnormal{ \ if } u \leq  \frac{\log(1-\mu)}{\log(1-\muab)} \, \kl(\muab|\mu)  \\
    2\e(1+\frac{\log(1-\muab)}{\log(1-\mu)}n) \e^{-n\frac{\log(1-\muab)}{\log(1-\mu)}u} \quad \textnormal{otherwise.}
\end{array} \right. 
$$
\end{proposition}
We know rewrite  equality (27) from \cite{honda2015imed} with our notations.
~\\ Let $ n \geq 1 $. We have from Proposition~\ref{prop : large devdev} that :
\beqan
&&\Esp_\nu\left[\ind_{\Set{\muhatab^n \leq \mu}} n \e^{n \kl(\muhatab^n|\mu)}\right] \\
&=& \int_0^\infty \Pr_\nu\left( \ind_{\Set{\muhatab^n \leq \mu}} n \e^{n \kl(\muhatab^n|\mu)} > x \right) \textnormal{ d }x \\
&=& \int_0^\infty \Pr_\nu\left(  n \e^{n \kl(\muhatab^n|\mu)} > x,\ \muhatab^n \leq \mu \right) \textnormal{ d}x \\
&=& \int_{-\infty}^\infty n^2 \e^{n u} \Pr_\nu\left(   \kl(\muhatab^n|\mu) > u,\ \muhatab^n \leq \mu \right) \textnormal{ d}u \qquad ( x = n \e^{n u} ,\ \textnormal{ d}x = n^2 \e^{n u} \textnormal{ d}u) \\
&=& \int_{-\infty}^{\frac{\kl(\muab|\mu)\log(1-\mu)}{\log(1-\muab)}} n^2 \e^{n u} \Pr_\nu\left(\kl(\muhatab^n|\mu) > u,\  \muhatab^n \leq \mu \right) \textnormal{ d}u \\
&+& \int_{\frac{\kl(\muab|\mu)\log(1-\mu)}{\log(1-\muab)}}^\infty n^2 \e^{n u} \Pr_\nu\left(   \kl(\muhatab^n|\mu) > u,\ \muhatab^n \leq \mu \right) \textnormal{ d}u \\
&\leq& \int_{-\infty}^{\frac{\kl(\muab|\mu)\log(1-\mu)}{\log(1-\muab)}} n^2 \e^{n u} \e^{-n \kl(\muab|\mu)} \textnormal{ d}u\\
&+& \int_{\frac{\kl(\muab|\mu)\log(1-\mu)}{\log(1-\muab)}}^\infty n^2 \e^{n u} 2\e(1+\frac{\log(1-\muab)}{\log(1-\mu)}n) \e^{-n\frac{\log(1-\muab)}{\log(1-\mu)}u} \textnormal{ d}u \ \ (\textnormal{Proposition \ref{prop : large devdev}}) \\
&=& n \e^{-n \kl(\muab|\mu)}\int_{-\infty}^{\frac{\kl(\muab|\mu)\log(1-\mu)}{\log(1-\muab)}} n \e^{n u}  \textnormal{ d}u \\
&+& 2n\e (1+\frac{\log(1-\muab)}{\log(1-\mu)}n) \int_{\frac{\kl(\muab|\mu)\log(1-\mu)}{\log(1-\muab)}}^\infty n \e^{- ( \frac{\log(1-\muab)}{\log(1-\mu)} - 1)n u}   \textnormal{ d}u \\
&=& n\e^{-n(1- \frac{\log(1-\mu)}{\log(1-\muab)}) \kl(\muab|\mu) } + 2n\e (1+\frac{\log(1-\muab)}{\log(1-\mu)}n) \dfrac{\e^{-n(1- \frac{\log(1-\mu)}{\log(1-\muab)}) \kl(\muab|\mu) }}{\frac{\log(1-\muab)}{\log(1-\mu)} - 1}\\
&=& \left(1+\dfrac{2e}{\frac{\log(1-\muab)}{\log(1-\mu)} - 1}\right) n \e^{-n(1- \frac{\log(1-\mu)}{\log(1-\muab)}) \kl(\muab|\mu) } \\
&+& \dfrac{2e}{1 - \frac{\log(1-\mu)}{\log(1-\muab)}}n^2 \e^{-n(1- \frac{\log(1-\mu)}{\log(1-\muab)}) \kl(\muab|\mu) }
\eeqan

~\\ To ends the proof, we use the following inequalities for $ r > 0 $:
\beqan
&& \sum\limits_{n\geq 1}n\e^{-n r} \leq \dfrac{1}{(1-\e^{-r})^2} \leq \dfrac{1}{(1-\e^{-r})^3}  \\
&& \sum\limits_{n\geq 1}n^2\e^{-n r} \leq  \dfrac{2}{(1-\e^{-r})^3} \,. \\
\eeqan
		
\end{proof}

\section{\label{app: extended proof}\texorpdfstring{\IMEDS, \IMEDStwo, \IMEDSstartwo}{TEXT}: Finite-time analysis}
In this subsection we rewrite and adapt the results established in Sections~\ref{app: IMEDSstar finite time analysis},\,\ref{app:proof_main_result} for \IMEDSstar strategy to the other considered strategies. Mainly, we rewrite the empirical lower bounds and upper bounds detailed in Lemmas~\ref{imedstar empirical lower bounds}, and~\ref{imedstar empirical upper bounds}. These inequalities form the basis of the analysis of \IMEDSstar strategy. For the sake of brevity and clarity, proofs are not given. 

\subsection{\IMEDS finite-time analysis}
Under \IMEDS strategy we do not solve empirical versions of optimisation problem~\ref{eq:lb_regret_structure} and pull the couples with minimum (pseudo) indexes. This leads to the following empirical bounds.
\begin{lemma}[\IMEDS: Empirical lower bounds]\label{imeds empirical lower bounds}Under \IMEDS, at each step time $t \!\geq\! 1$, for all couple $(a,b) \!\notin\! \Ohat^\star(t)$,
		\[ 
         \log\left(N_{a_{t+1}, b_{t+1}}(t)\right)  \leq \sum\limits_{b'\in \Bhat_{a,b}(t)} N_{a,b'}(t) \, \klof{\muhat_{a,b'}(t)}{\muhat_{b}^\star(t) - \omegabb}   + \log\left(N_{a,b'}(t)\right) \,.
        \]
    Furthermore, for all couple $(a,b) \!\in\! \Ohat^\star(t)$,
        \[
          N_{a_{t+1},b_{t+1}}(t) \leq N_{a,b}(t) \,.
        \]
\end{lemma}

\begin{lemma}[\IMEDS: Empirical upper bounds]\label{imeds empirical upper bounds} Under \IMEDS, at each step time $t \!\geq\! 1$ such that $(a_{t+1},b_{t+1}) \!\notin\! \Ohat^\star(t)$ we have
		\[
		\sum\limits_{b'\in \Bhat_{a_{t+1},b_{t+1}}(t)} N_{a_{t+1},b'}(t)\, \klof{\muhat_{a_{t+1},b'}(t)}{\muhat_{b_{t+1}}^\star(t) - \omega_{b_{t+1},b'}} \leq \log\!\left(N_{b_{t+1}}(t)\right) \,.
		\]
		In particular
		\[
		N_{a_{t+1},b_{t+1}}(t)\,\klof{\muhat_{a_{t+1},b_{t+1}}(t)}{\muhat_{b_{t+1}}^\star(t)} \leq \log\!\left(N_{b_{t+1}}(t)\right) \,.
		\]
\end{lemma}
Based on this lemmas, one can prove \IMEDS Pareto-optimality in a similar way as for \IMEDSstar strategy.
\begin{proposition}[\IMEDS: Upper bounds ] \label{prop:imeds upper bounds} Let $ \nu \!\in\! \cD_\omega$. Let $ 0 \!<\! \epsilon \!<\! \epsilon_\nu$ and $\gamma \!\in\! (0,1/2)$. Let us introduce 
\[
 \cT_{\epsilon, \gamma} \coloneqq \Set{t \geq 1:\ \begin{array}{l}
		  \Ohat^\star(t) = \cO^\star  \\  \forall (a,b) \textnormal{ s.t. } N_{a,b}(t) \geq \gamma\, N_{a_{t+1},b_{t+1}}(t) \textnormal{ or } (a,b) \in \cO^\star, \ \abs{\muhatab(t) - \muab } < \epsilon 
		\end{array}   }\,. 
\] 
Then under \IMEDS strategy,
\[
\Esp_\nu[\abs{\cT_{\epsilon,\gamma}^{ c}}] \leq \dfrac{17}{\gamma\epsilon^4}  \abs{\cA}^2\abs{\cB}^2 +  2\abs{\cA}^2\abs{\cB}(1 + E_\nu)^{\abs{\cB}} 
\]
and for all horizon time $ T \!\geq\! 1$, for all arm $a \in \cA$,
\beqan 
 \min\limits_{b:\, (a,b) \notin \cO^\star } \dfrac{1}{\log\!\left(N_b(T)\right)}\sum\limits_{b'\in \Bab}N_{a,b'}(T)\, \kl(\mu_{a,b'}|\mu_b^\star - \omegabb) 
&\leq& (1 + \alpha_\nu(\epsilon))\left[ 1 + \gamma \dfrac{M_\nu}{m_\nu} \right]  \\
&+&\dfrac{M_\nu\abs{\cT_{\epsilon,\gamma}}}{\min_{b\in\cB}\log\!\left(N_b(T)\right)}   
\eeqan
where $m_\nu$ and $M_\nu$ are defined  as follows:
\[
m_\nu = \min\limits_{\substack{(a,b) \notin \cO^\star \\ b'\in \Bab}} \klabb,\qquad M_\nu = \max\limits_{(a,b) \notin \cO^\star}\sum\limits_{b'\in \Bab} \klabb  \,.
\]
Furthermore, we have 
\[ \forall (a,b) \notin \cO^\star,\quad \dfrac{N_{a,b}(T)}{\log\!\left(N_b(T)\right)} \leq \dfrac{1 + \alpha_\nu(\epsilon)}{\klab} + \dfrac{\abs{\cT_{\epsilon,\gamma}^c}}{\log\!\left(N_b(T)\right)}  \,.
\]
Refer to Appendix~\ref{notations} for the definitions of $\epsilon_\nu$, $\alpha_\nu(\cdot)$ and $E_\nu$.
\end{proposition}
From the previous proposition, we deduce the following corollary by doing $T \!\to\! \infty$, then $\epsilon, \gamma \!\to\! 0$.
\begin{corollary}[\IMEDS: Pareto optimality]\label{cor: imeds pareto optimality}Let $\nu \in \cD_\omega$. Under \IMEDS strategy we have
$$ \forall a \in \cA,\quad \limsupT\min\limits_{b:\, (a,b) \notin \cO^\star } \dfrac{1}{\log\!\left(N_b(T)\right)} \sum\limits_{b'\in \Bab}N_{a,b'}(T) \kl(\mu_{a,b'}|\mu_b^\star - \omegabb)  \leq 1 \,.$$
\end{corollary}

\subsection{\textit{Uncontrolled scenario}: Finite-time analysis}
When \textit{uncontrolled scenario} is considered, the learner does not choose the users to deal with and the exploration phases may be performed with some delay. This can be formalized within the empirical bounds induced by \IMEDStwo and \IMEDSstartwo strategies.  
\subsubsection{Empirical bounds on the numbers of pulls}
For time step $ t\geq 1$, let us introduce the last return time of couple $(a_{t+1},b_{t+1}) \!\in\!\cB$ as 
\[ \tau_t \coloneqq \min \Set{ t - t':\ b_{t'} = b_{t+1},\ t \geq t' \geq 1 } \,.
\]
By definition of $\tau_t$ we have 
\[ N_{a_{t+1},b_{t+1}}(t) = N_{a_{t+1},b_{t+1}}(t -\tau_t)  \,. \]
Now, empirical bounds on the numbers of pull can be formulated for the \textit{uncontrolled scenario}. These inequalities are the same as those  for the \textit{controlled scenario} up to (random) time-delays.
\begin{lemma}[\textit{Uncontrolled scenario}: Empirical lower bounds]\label{scenario 2: empirical lower bounds}Under \IMEDStwo and \IMEDSstartwo, at each step time $t \!\geq\! 1$ there exists a random time delay $\sigma_t$ such that $0 \!\leq \! \sigma_t \!\leq\! \tau_t$ and for all couple $(a,b) \!\notin\! \Ohat^\star(t\!-\!\sigma_t)$
		\[ 
         \log\!\left(N_{a_{t+1}, b_{t+1}}(t\!-\!\sigma_t)\right)  \!\leq\!\!\!\!\! \sum\limits_{b'\in \Bhat_{a,b}(t\!-\!\sigma_t)}\!\!\!\!\! N_{a,b'}(t\!-\!\sigma_t) \, \klof{\muhat_{a,b'}(t\!-\!\sigma_t)}{\muhat_{b}^\star(t\!-\!\sigma_t) \!-\! \omegabb}   \!+\! \log\!\left(N_{a,b'}(t\!-\!\sigma_t)\right).
        \]
    Furthermore, for all couple $(a,b) \!\in\! \Ohat^\star(t\!-\!\sigma_t)$,
        \[ N_{a_{t+1},b_{t+1}}(t\!-\!\sigma_t) \leq N_{a,b}(t\!-\!\sigma_t) 
           \,.
        \]
	\end{lemma}
\begin{lemma}[Empirical upper bounds]\label{imedstar2 empirical upper bounds}
		Under \IMEDStwo and \IMEDSstartwo, at each step time $t \!\geq\! 1$ such that $(a_{t+1},b_{t+1}) \!\notin\! \Ohat^\star(t)$, we have
		\[
		\sum\limits_{b'\in \Bhat_{a_{t+1},\bul_t}(t\!-\!\sigma_t)} N_{a_{t+1},b'}(t\!-\!\sigma_t)\, \klof{\muhat_{a_{t+1},b'}(t\!-\!\sigma_t)}{\muhat_{\bul_t}^\star(t\!-\!\sigma_t) - \omega_{\bul_t,b'}} \leq \log\!\left(N_b(t\!-\!\sigma_t)\right) \,,
		\]
		where $\sigma_t$ is a random time delay  such that $0 \!\leq \! \sigma_t \!\leq\! \tau_t$.
		Furthermore, we have under \IMEDStwo  
		\[b_{t+1} = \bul_t,\qquad
			N_{a_{t+1},b_{t+1}}(t\!-\!\sigma_t)\,\klof{\muhat_{a_{t+1},b_{t+1}}(t\!-\!\sigma_t)}{\muhat_{b_{t+1}}^\star(t\!-\!\sigma_t)} \leq \log\!\left(N_{b_{t+1}}(t\!-\!\sigma_t)\right)
        \]
        and under \IMEDSstartwo
		\[
		\dfrac{N_{a_{t+1},b_{t+1}}(t\!-\!\sigma_t)}{\log(t\!-\!\sigma_t)} \!\leq\!\! \left\{ \begin{array}{l}
		    \!\!\!\!\dfrac{1}{\klof{\muhat_{a_{t+1},\bul_t}(t\!-\!\sigma_t)}{\muhat_{\bul_t}^\star(t\!-\!\sigma_t)}}  \textnormal{ , if } c_{a_{t+1}}(t\!-\!\sigma_t) = c_{a_{t+1}}^+(t\!-\!\sigma_t) \\
		    \!\!\!\!\min\!\left(\dfrac{1}{\klof{\muhat_{a_{t+1},b_{t+1}}(t\!-\!\sigma_t)}{\muhat_{\bul_t}^\star(t\!-\!\sigma_t) - \omega_{\bul_t,b_{t+1}}}}, \nopt_{a_{t+1},b_{t+1}}(t\!-\!\sigma_t)\!\! \right)\!\!\! \ \textnormal{, else.} 
		\end{array} \right. 
		\]
\end{lemma}
Thus, we prove respectively the Pareto-optimality and optimality of \IMEDStwo and \IMEDSstartwo since we show that the empirical means $\muhatab(t\!-\!\sigma_t)$ of couples $(a,b)$ involved in the previous inequalities concentrate as in the case of the \textit{controlled scenario}. This is the case as it is stated in Lemmas~\ref{imedstar2 concentration} of the next subsection.
\subsubsection{Concentration inequality with bounded time delays}
We prove a concentration lemma that does not depend on the followed strategy. It is a rewritting for the case of \textit{controlled scenario} of Lemma~\ref{imedstar concentration}.
	\begin{lemma}[Concentration inequalities]\label{imedstar2 concentration} Let  $\nu \!\in\! \cD_\omega $, $  0 \!<\! \epsilon$, $ \gamma \!\leq\! 1/2 $ and $(a,b),\, (a',b') \!\in\! \cA\!\times\!\cB$. Then for all sequence of stopping times $(\sigma_t)_{t\geq1}$ such that $0 \!\leq\sigma_t\leq\tau_t$ for all $t\!\geq\!1$, we have 
		\[
		\Esp_\nu\!\left[\sum\limits_{t\geq 1}\ind_{\Set{(a_{t+1},b_{t+1})=(a,b),\ N_{a',b'}(t-\sigma_t) \geq \gamma  N_{a,b}(t-\sigma_t),\  \abs{\muhat_{a',b'}(t-\sigma_t) - \mu_{a',b'}} \geq \epsilon }} \right] \leq \dfrac{17}{\gamma\epsilon^4} \,.
		\]
	\end{lemma}
	\begin{remark} There is no need to adapt Lemma~\ref{imedstar large deviation} for the case of \textit{controlled scenario} since this concentration lemma does not involve the current time steps explicitly.
    \end{remark}
\begin{proof} It is pointed out that for all time step $t \!\geq\! 1$, $N_{a_{t+1},b_{t+1}}(t\!-\!\sigma_t)\!\geq\! N_{a_{t+1},b_{t+1}}(t\!-\!\tau_t) \!=\! N_{a_{t+1},b_{t+1}}(t)$, then we proceed as in Appendix~\ref{app: concentration_lemmas}.

~\\Considering the stopping times $ \tau_{a,b}^n = \inf{ \Set{t \geq 1, N_{a,b}(t) = n} }$  we will rewrite the sum  
		$$ \sum\limits_{t\geq 1}\ind_{\Set{(a_{t+1},b_{t+1})=(a,b),\ N_{a',b'}(t-\sigma_t) \geq \gamma  N_{a,b}(t-\sigma_t),\  \abs{\muhat_{a',b'}(t-\sigma_t) - \mu_{a',b'}} \geq \epsilon }} $$ and use an Hoeffding's type argument.
		\beqan
		&& \sum\limits_{t \geq 1}\ind_{\Set{(a_{t+1},b_{t+1})=(a,b),\  N_{a',b'}(t-\sigma_t) \geq \gamma N_{a,b}(t-\sigma_t),\  \abs{\muhat_{a',b'}(t-\sigma_t) - \mu_{a',b'}} \geq \epsilon }} \\
		&\leq& \sum\limits_{t \geq 1}\ind_{\Set{(a_{t+1},b_{t+1})=(a,b),\  N_{a',b'}(t-\sigma_t) \geq \gamma N_{a,b}(t),\  \abs{\muhat_{a',b'}(t-\sigma_t) - \mu_{a',b'}} \geq \epsilon }} \\
		&\leq& \sum\limits_{t \geq 1}\sum\limits_{n \geq 1,\, m \geq 0}\ind_{\Set{ \tau_{a,b}^n = t +1, N_{a',b'}(t-\sigma_t) = m }} \ind_{\Set{ m \geq  \gamma(n-1) ,\  \abs{\muhat_{a',b'}^{m} - \mu_{a',b'}} \geq \epsilon }} \\
		&=& \sum\limits_{ m \geq 0} \sum\limits_{n \geq 1}\ind_{\Set{ m \geq \gamma(n-1),\  \abs{\muhat_{a',b'}^{m} - \mu_{a',b'}} \geq \epsilon }}\sum\limits_{t \geq 1}\ind_{\Set{ \tau_{a,b}^n = t +1, N_{a',b'}(t-\sigma_t) = m }}  \\
		&\leq& \sum\limits_{ m \geq 0} \sum\limits_{n \geq 1}\ind_{\Set{ m \geq \gamma(n-1),\  \abs{\muhat_{a',b'}^{m} - \mu_{a',b'}} \geq \epsilon }}\sum\limits_{t \geq 1}\ind_{\Set{ \tau_{a,b}^n = t +1}}  \\
		&\leq& \sum\limits_{ m \geq 0} \sum\limits_{n \geq 1}\ind_{\Set{ m \geq  \gamma(n-1),\  \abs{\muhat_{a',b'}^{m} - \mu_{a',b'}} \geq \epsilon }} \,, \\
		\eeqan 
	where the $\muhat_{a',b'}^{m} $ are defined in  Appendix~\ref{notations}. The proof ends the same way as in Appendix~\ref{app: concentration_lemmas}.
\end{proof}

\section{Continuity of solutions to parametric linear programs}
In this section we recall Lemma 13  established in \citet{magureanu2014oslb} on the continuity of solutions to parametric linear programs.
\begin{lemma} \label{lemma : continuity lp} Consider $ K \in \Real_+^{B\times B} $, $ \Delta \in \Real_+^B  $, and $ \cH \subset \Real_+^{B\times B}\times\Real_+^B $. Define $h = (K,\Delta)$. Consider the function $ Q$ and the set-valued map $Q^\star$
$$ Q(h) = \inf\limits_{x \in \Real_+^B}\Set{\Delta\cdot x | K \cdot x \geq 1 } $$
$$ Q^\star(h) = \Set{x\geq0 : \Delta \cdot x \leq Q(h) | K \cdot x \geq 1 } \,.$$
Assume that:
~\\ (i) For all $h \in \cH$, all rows and columns of $K$ are non-identically 0
~\\ (ii) $\min\limits_{h  \in \cH}\min\limits_{b \in B}\Delta_b > 0 .$
~\\ Then:
~\\ (a) $Q$ is continuous on $ \cH$
~\\ (b) $Q^\star$ is upper hemicontinuous on $ \cH$.
\end{lemma}

\section{Details on numerical experiments}
\label{app:detail_numerical_exp}


For the fixed configuration experiments we used the weight matrix $\omega$ of Table~\ref{tab:w_fixed} and the configuration $\nu$ described in Table~\ref{tab:nu_fixed}. $\omega$ and $\nu$ have been chosen at random in such a way that the regret under \IMED exceeds the structured lower bound on the regret. This means the structure $\omega$ is informative for the bandit configuration $\nu$ and not taking it into account hinders optimality. 

\begin{table}[!ht]
\centering
\begin{tabular}{|c|c|c| c|c|c| c|c|c| c|c|}
\hline
user\textbackslash user & $b_1$ & $b_2$ & $b_3$ & $b_4$ & $b_5$ & $b_6$ & $b_7$ & $b_8$  & $b_9$ & $b_{10}$ \\
\hline

$b_1$ & 0 & 0.07 & 0.07 & 0.12 & 0.20 & 0.05 & 0.16 & 0.14 & 0.28 & 0.03 \\
\hline

$b_2$ & 0.07 & 0 & 0.14 & 0.13 & 0.21 & 0.12 & 0.09 & 0.07 & 0.21 & 0.04 \\
\hline

$b_3$ & 0.07 & 0.14 & 0 & 0.19 & 0.27 & 0.12 & 0.11 & 0.13 & 0.25 & 0.10 \\
\hline

$b_4$ & 0.12 & 0.13 & 0.19 & 0 & 0.26 & 0.17 & 0.22 & 0.20 & 0.34 & 0.09 \\
\hline

$b_5$ & 0.20 & 0.21 & 0.27 & 0.26 & 0 & 0.25 & 0.18 & 0.20 & 0.32 & 0.17 \\
\hline

$b_6$ & 0.05 & 0.12 & 0.12 & 0.17 & 0.25 & 0 & 0.21 & 0.19 & 0.33 & 0.08 \\
\hline

$b_7$ & 0.16 & 0.09 & 0.11 & 0.22 & 0.18 & 0.21 & 0 & 0.02 & 0.14 & 0.13 \\
\hline

$b_8$ & 0.14 & 0.07 & 0.13 & 0.20 & 0.20 & 0.19 & 0.02 & 0 & 0.16 & 0.11 \\
\hline

$b_9$ & 0.28 & 0.21 & 0.25 & 0.34 & 0.32 & 0.33 & 0.14 & 0.16 & 0 &  0.25\\
\hline

$b_{10}$ & 0.03 & 0.04 & 0.10 & 0.09 & 0.17 & 0.08 & 0.13 & 0.11 & 0.25 & 0 \\ 
\hline

\end{tabular}
    \caption{weight matrix $\omega$ used in the fixed configuration experiment.}
    \label{tab:w_fixed}
\end{table}

\begin{table}[H]
\centering
\begin{tabular}{|c|c|c|c|c|c|c|c|c|c|c|}
\hline
  arm \textbackslash user & $b_1$ & $b_2$ & $b_3$ & $b_4$ & $b_5$ & $b_6$ & $b_7$ & $b_8$  & $b_9$ & $b_{10}$ \\
\hline
    $a_1$ & 0.15 & 0.11 & 0.19 & 0.19 & 0.08 & 0.15 & 0.09 & 0.08 & 0.13 & 0.13 \\
    \hline
    
    $a_2$ & 0.70 & 0.73 & 0.71 & 0.71 & 0.70 & 0.71 & 0.78 & 0.79 & 0.64 & 0.70 \\
     \hline
    
    $a_3$ & 0.13 & 0.14 & 0.13 & 0.15 & 0.02 & 0.08 & 0.17 & 0.17 & 0.04 & 0.14 \\
    \hline
    
    $a_4$ & 0.02 & 0.04 & 0.09 & 0.05 & 0.16 & 0.02 & 0.11 & 0.11 & 0.06 & 0.01 \\
    \hline
    
    $a_5$ & 0.95 & 0.98 & 1.00 & 0.97 & 0.98 & 0.98 & 0.90 & 0.91 & 0.84 & 0.97 \\
\hline
\end{tabular}

    \caption{configuration $\nu$ used in the fixed configuration experiment.}
    \label{tab:nu_fixed}
\end{table}

\begin{figure}[H]
    \centering
     \includegraphics[scale= 0.8]{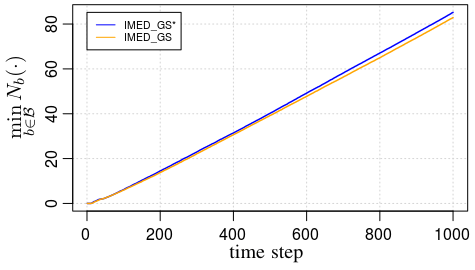} 
     \caption{ $\min_{b\in\cB}N_b(\cdot)$ approximated over $1000$ runs. At each run we sample uniformly at random  
a weight matrix $\omega$  and then sample uniformly at random a configuration $\nu \!\in\!\cD_\omega$. }
 \label{fig:experiments minNb}
\end{figure}

\end{document}